\documentclass[conference]{IEEEtran}
\usepackage{todonotes}
\IEEEoverridecommandlockouts
\usepackage{cite}
\usepackage{amsmath,amssymb,amsfonts}
\usepackage{algorithmic}
\usepackage{graphicx}
\usepackage{textcomp}
\usepackage{xcolor}
\usepackage{booktabs}
\def\model{ADAPTD}
\def\Chain{Partitioned belief filtering}
\def\chain{partitioned belief filtering}
\def\BibTeX{{\rm B\kern-.05em{\sc i\kern-.025em b}\kern-.08em
    T\kern-.1667em\lower.7ex\hbox{E}\kern-.125emX}}

\DeclareMathOperator*{\argmax}{arg\,max}
\usepackage{xurl}

\usepackage[utf8]{inputenc}
\usepackage[T1]{fontenc}
\usepackage{bbm}
\usepackage[titlenumbered,linesnumbered ,ruled]{algorithm2e}
\usepackage{amsthm}

\usepackage{booktabs}
\usepackage{mathtools}
\usepackage{xcolor}
\usepackage{float}
\usepackage{tikz}  
\usepackage{textcomp}
\usepackage[mathscr]{euscript}
\usepackage{acro}
\usepackage[framemethod=tikz]{mdframed}
\usepackage{lipsum}
\mathtoolsset{showonlyrefs}
\usepackage{multirow}
\DeclareSymbolFont{rsfs}{U}{rsfs}{m}{n}
\DeclareSymbolFontAlphabet{\mathscrsfs}{rsfs}

\newtheorem{proposition}{\textbf{Proposition}}
\newtheorem{theorem}{\textbf{Theorem}}

\newtheorem{assumption}{\textbf{Assumption}}
\newtheorem{definition}{\textbf{Definition}}
\DeclareAcronym{MR}{
  short = MR,
  long  = max ratio,
  long-plural = s
}
\DeclareAcronym{SOC}{
  short = SOC,
  long  = security operations center,
  long-plural = s
}
\DeclareAcronym{POMDP}{
  short = POMDP,
  long  = partially observable Markov decision process,
  long-plural = es
}
\DeclareAcronym{MKL}{
  short = MKL,
  long  = max KL,
  long-plural = s
}
\DeclareAcronym{SA}{
  short = SA,
  long  = situational awareness,
  long-plural = es
}
\DeclareAcronym{IDS}{
  short = IDS,
  long  = intrusion detection system,
  long-plural = s
}
\DeclareAcronym{KL}{
  short = KL,
  long  = Kullback–Leibler,
  long-plural = 
}
\DeclareAcronym{DoS}{
  short = DoS,
  long  = denial of service,
  long-plural = s
}
\DeclareAcronym{MTTR}{
  short = MTTR,
  long  = mean time to remediate,
  long-plural = 
}
\DeclareAcronym{HMM}{
  short = HMM,
  long  = hidden Markov model,
  long-plural = s
}
\DeclareAcronym{ML}{
  short = ML,
  long  = machine learning,
  long-plural = s
}
\DeclareAcronym{IRS}{
  short = IRS,
  long  = intrusion response system,
  long-plural = s
}
\DeclareAcronym{CUSUM}{
  short = CUSUM,
  long  = cumulative sum control chart,
  long-plural = s
}
\DeclareAcronym{MSGPRT}{
short= MSGPRT,
long = multihypothesis sequential generalized probability
ratio test,
long-plural =s
}
\DeclareAcronym{BSPRT}{
short= BSPRT,
long = binary hypothesis sequential probability
ratio test,
long-plural =s
}
\DeclareAcronym{CHT}{
short= CHT,
long = composite hypothesis test
ratio test,
long-plural =s
}
\DeclareAcronym{ACL}{
short= ACL,
long = access control list,
long-plural =s
}
\DeclareAcronym{GLRT}{
short= GLRT,
long = generalized likelihood ratio test,
long-plural =s
}

\DeclareAcronym{MTBFD}{
short= MTBFD,
long = mean time between false detections,
long-plural =s
}
\DeclareAcronym{APT}{
short= APT,
long = advanced persistent threat,
long-plural =s
}
\DeclareAcronym{MC}{
short = MC,
long = Monte Carlo,
long-plural = s
}

\def\stage{s}

\def\stageRv{S}
\def\stageSet{\mathcal{S}}

\def\iSet{\mathcal{I}}
\def\jSet{\mathcal{J}}
\def\kSet{\mathcal{K}}

\def\E{\mathbb{E}}
\def\P{\mathbb{P}}
\def\alertVector{y}

\def\alertVectorRv{Y}

\def\hypothesis{h}
\def\hypothesisSet{\mathcal{H}}

\def\policy{\rho}
\def\policySet{\mathcal{P}}
\def\dSet{\mathcal{D}}
\def\dRv{D}
\def\powerSet{\mathscrsfs{P}}
\usepackage{comment}

\def\hypothesisRv{H}
\def\tpRate{\delta}
\def\fpRate{\zeta}
\def\tranProb{a}

\def\prior{q}
\def\belief{\pi}
\def\beliefRv{\Pi}
\def\distributedBelief{\phi}

\def\threshold{\theta}

\def\hSet{\mathcal{H}}

\newcommand\myleqA{\stackrel{\mathclap{\normalfont\mbox{(a)}}}{\leq}}
\newcommand\myleqB{\stackrel{\mathclap{\normalfont\mbox{(b)}}}{\leq}}
\usepackage{siunitx}
\usepackage[normalem]{ulem}

\newcommand{\revision}[1]{\textcolor{black}{#1}}
\newcommand{\tcnsrevision}[1]{\textcolor{black}{#1}}
\begin{document}

\title{{\model}: Adaptive Detection and Proactive Threat Defense for Autonomous APT Attacks}


\author{\IEEEauthorblockN{Yeongwoo Kim\IEEEauthorrefmark{1},  Quanyan Zhu\IEEEauthorrefmark{2} and Gy\"orgy D\'an\IEEEauthorrefmark{1}}
\IEEEauthorblockN{\IEEEauthorblockA{\IEEEauthorrefmark{1} Department of Network and Systems Engineering,
KTH Royal Institute of Technology, Stockholm, Sweden\\ \IEEEauthorrefmark{2} Department of Electrical and
Computer Engineering, New York University, New York, USA\\
Email: \IEEEauthorrefmark{1}\{yeongwoo, gyuri\}@kth.se, \IEEEauthorrefmark{2}qz494@nyu.edu}
}
}
\maketitle

\begin{abstract}
Advanced persistent threat (APT) actors increasingly employ sophisticated techniques to propagate laterally through segmented enterprise networks. Timely detection and defense depend on cross-subnetwork coordination, yet maintaining global situational awareness generates substantial communication overhead. To manage this tradeoff, flexible monitoring and adaptable containment are imperative. This paper presents ADAPTD, a communication- and computation-efficient, decision-theoretic framework integrating: (i) compact kill chains for identifying diverse attack vectors, (ii) an immediate blocking mechanism for timely containment, and (iii) a predictive eviction strategy to restore system security. Our experiments validate ADAPTD's effectiveness across diverse threat scenarios. First, our decentralized belief update scheme outperforms state-of-the-art diffusion HMM. Second, ADAPTD substantially reduces false evictions compared to transformer-based detection. Third, under noisy environments, adaptive blocking contains attackers while minimizing unnecessary disruption. Lastly, the ablation study confirms that combining two defensive actions significantly reduces the defender's total cost.

\end{abstract}

\section{Introduction}
\label{sec:intro}
    
    

An \ac{APT} is a threat actor that employs various techniques and tools to gain illegitimate access to networked systems, often over an extended period of time~\cite{hu2015dynamic}.  \ac{APT} may involve human threat actors, but there have also been cases of self-propagating malware (e.g., Stuxnet autonomously compromised and destroyed nuclear facilities), causing significant financial loss to organizations~\cite{farwell2011stuxnet}. Such autonomous threats have then been reported in critical infrastructures, such as financial institutions~\cite{cimpanu2020chilean}, government agencies~\cite{abrams2020netwalker}, the healthcare industry~\cite{collier2020major}, and the energy sector~\cite{tsvetanov2021effect}.

A key challenge posed by such attacks is their ability to exploit diverse attack vectors ranging from social engineering to USB flash drives~\cite{zimba2017malware}, enabling attacks even against \revision{well protected corporate infrastructures}. Once compromised, an entire networked system can be infected autonomously through lateral movement via connected subnetworks, without the need to communicate with a command-and-control server, i.e., without any form of remote coordination. 
The diversity of attack vectors implies that, from the perspective of the defender, the starting point of an attack is hard to predict. As attacks propagate, defenders have to i) account for all possible initial attack vectors and ii) strategically combine diverse defensive actions to prevent the compromise of critical assets, while maintaining a certain level of business continuity.

\begin{figure}[!t]
\begin{center}
  \includegraphics[width=0.9\linewidth]{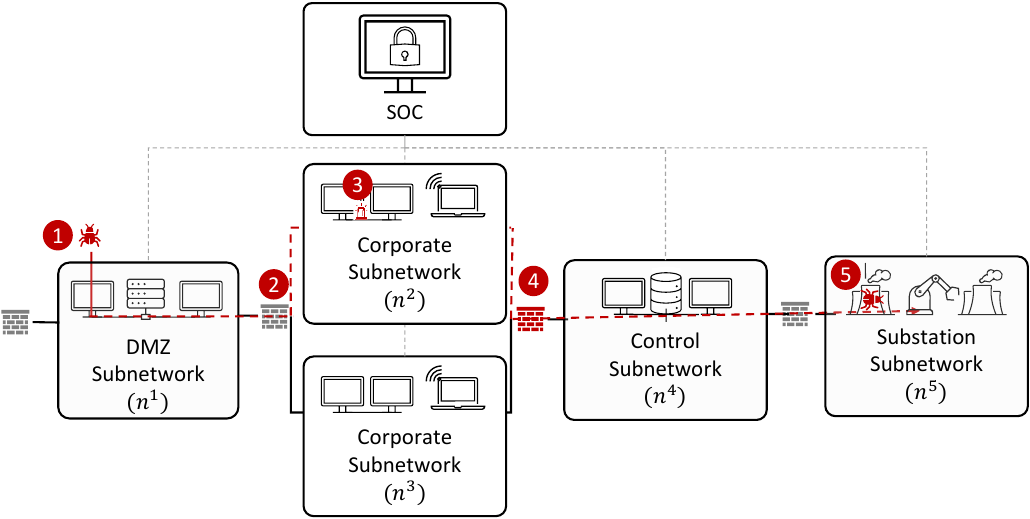}\vspace{-4mm}
  \caption{Illustration of diverse attack vectors of a ransomware that autonomously propagates to a targeted critical asset. In the example, 1. The malware compromises the DMZ subnetwork by using unauthorized USB access. 2. It penetrates the firewalls between DMZ and control subnetworks. 3. The corporate network detects the attack. 4. The detection leads to the blocking of all traffic from DMZ and corporate subnetworks. 5. The attack compromises the substation subnetwork.}
  \label{fig:first_page}
\end{center}
\vspace{-5mm}
\end{figure}

Effective threat response relies on real-time \ac{SA}, which depends on efficient processing of alerts triggered by \acp{IDS}. In centralized approaches, beliefs are computed by aggregating alerts at the machine or host level~\cite{holgado2017real, ourston2003applications}. 
\revision{Concurrently, another line of work utilizing provenance- or behavior-based detection approaches (e.g., HOLMES~\cite{milajerdi2019holmes}) has demonstrated remarkable capabilities in detecting APTs by aggregating system events.}
However, this aggregation introduces substantial communication overhead, \revision{underscoring the need for a bandwidth-efficient belief update mechanism. Designing such a decentralized scheme, however, poses a significant technical challenge, as maintaining global \ac{SA} without incurring prohibitive communication cost among distinct subnetworks is inherently difficult.} 


\ac{SA} is a prerequisite for threat response, but it is not sufficient on its own. Given a belief about the attacker's state, existing \acp{IRS} typically deploy lightweight defensive actions, such as revoking credentials or blocking access, to delay the attacker's progression~\cite{miehling2018pomdp,xiao2018dynamic,hammar2022intrusion}. While such actions are attractive due to their low cost and limited disruption, relying on them exclusively can lead to prolonged engagement, eventual compromise of critical assets, or continuous defensive overhead that impacts business continuity. To address this, effective threat response must combine lightweight actions with heavyweight defenses, such as evicting compromised components. However, integrating these two classes of defenses in a principled and adaptive manner, particularly under uncertainty about the attacker's progression, remains an open \revision{decision-theoretic} problem.

Fig.~\ref{fig:first_page} shows an attack scenario in which an autonomous APT compromises the demilitarized zone (DMZ), and subsequently it compromises the substation subnetwork. The defender could have protected the substation subnetwork from the attacker if it had blocked outgoing traffic from the corporate subnetwork in a timely manner. Nonetheless, to do so with minimal disruption to business continuity, the defender needs \ac{SA}, i.e., an accurate belief about the attack's progression, as well as a \revision{cost-aware decision rule} for when to isolate subnetworks and when to shut down parts of the system for evicting the attacker.

In this paper, we propose \revision{a decision-theoretic defense framework with decentralized belief update}, called {\model}, that i) decomposes large-scale kill chains into smaller, subnetwork-specific chains; ii) blocks lateral movement for attack mitigation and for more accurate decision-making; and iii) estimates the future cost of mitigation actions for deciding when to evict the attacker from the system.
\revision{Unlike provenance- or behavior-based detection systems (e.g., HOLMES~\cite{milajerdi2019holmes}) that aim to generate high-level alerts using low-level system events, {\model} operates as a decision-theoretic engine that ingests these high-level alerts and optimizes proactive defense under system-wide resource constraints. Consequently, the primary objective of our work is to establish a rigorous framework for defensive decision-making under uncertainty, rather than refining passive detection algorithms. Given this decision-theoretic focus on the systemic evaluation of dynamic defense policies, comprehensive numerical simulations serve as a rigorous environment to mathematically validate our framework across diverse attacker behaviors and system sizes.}

The contributions of our work are as follows:  
\begin{itemize}
\item \textbf{Decentralized situational awareness}: We propose to decompose the attacker's kill chain into subnetwork-specific kill chains, allowing each subnetwork to maintain its own belief state based on local alerts and share it only with neighboring subnetworks. This decentralized approach reduces communication overhead compared to a centralized approach at comparable estimation accuracy.
\item \textbf{Proactive Defense}: Our proposed framework proactively isolates suspicious subnetworks to contain the propagation of APTs until the attack is confirmed. Outgoing traffic from subnetworks is blocked based on  the belief about the attacker's progression taking into account its impact on business continuity. 
\item \textbf{Minimum Cost Eviction}: \revision{We employ \ac{MC} simulations based on the belief to find the optimal time for evicting the attacker from the system taking into account the expected future cost of blocking actions.}
\item \textbf{\revision{Numerical Evaluation}}: We use \revision{extensive numerical simulations to evaluate the proposed policies against state-of-the-art baselines in segmented environments}. Our results show that ADAPTD outperforms the baselines in terms of defense cost, demonstrates robust performance against diverse attack vectors, and has lower  computational and communication overhead.  
\end{itemize}

\noindent The rest of the paper is organized as follows. We discuss related work in Section \ref{sec:related_work}. We formulate the system model in Section \ref{sec:attack_model} and present the {\model} in Section~\ref{sec:adapti}. Section~\ref{sec:simulation} provides an evaluation of {\model}, and Section~\ref{sec:conclusion} concludes the paper.\vspace{-4mm}

\section{Related Work} 
\label{sec:related_work}
\Acp{IDS} generate alerts to identify the attacker's malicious activities~\cite{porras1992penetration, lee1999data, gu2007bothunter, pei2016hercule, hossain2017sleuth, milajerdi2019holmes, alrubayyi2021novel,kim2022active}. Phillip et al. used signatures of known attacks to generate alerts~\cite{porras1992penetration}, which, however, cannot detect unknown attacks. The authors of~\cite{lee1999data} used a data mining algorithm to learn rules to detect anomalous activities. However, this approach raises numerous false alerts, causing alert fatigue. In~\cite{gu2007bothunter}, the authors correlated communication between internal and external devices to reduce the number of alerts. Pei et al. correlated dispersed malicious activities in log data using community detection~\cite{pei2016hercule}. The authors in~\cite{hossain2017sleuth, milajerdi2019holmes} constructed a graph to correlate suspicious system calls, aiming to detect \ac{APT} attacks. Alrubayyi et al. used a negative-positive-selection algorithm to detect the novel malware attack~\cite{alrubayyi2021novel}. Kim et al. considered the investigation of alerts with dynamic prioritization of alerts~\cite{kim2022active}. 
\revision{While these state-of-the-art tools achieve high detection accuracy, they operate as passive monitors and lack a cost-aware mechanism for active defense. Therefore, they are complementary to our framework; they serve as upstream alert generators, whereas {\model} functions as a top-level decision-theoretic engine that optimizes mitigation actions under uncertainty.}

\Acp{IRS} choose defense actions autonomously upon detecting malicious activity~\cite{kreidl2004feedback, zonouz2013rre, iannucci2016high, iannucci2016probabilistic, miehling2015optimal, miehling2018pomdp,xiao2018dynamic, hammar2022intrusion, shen2026privacy, anwar2022honeypot, peng2025propagation, jia2025novel}. Xiao et al.~\cite{xiao2018dynamic} formulated a Stackelberg game to choose lightweight defensive actions aiming to mitigate malware propagation, but do not consider that malware may use different attack vectors.  Kreidl et al.~\cite{kreidl2004feedback} used a partially observable Markov decision process for maintaining \ac{SA} regarding the security state of the host and to select defense actions, considering the defender's cost~\cite{miehling2018pomdp}. The authors in~\cite{zonouz2013rre} developed a hierarchical approach comprised of local and global engines. The local engine maintains the \ac{SA} on an attack-response chain, while the global engine aggregates the \ac{SA} from local engines and chooses the defense actions. Iannucci et al.~\cite{iannucci2016high, iannucci2016probabilistic} chose a response action to counteract the attacks, considering a response that may mitigate multiple attacks. Miehling et al.~\cite{miehling2015optimal,miehling2018pomdp} used Bayesian attack graphs to choose defensive actions where the action changes the attacker's available exploits, improving SA. 

\tcnsrevision{Considering strategic attackers~\cite{jia2025novel,singh2022dynamic}, Jia et al.~\cite{jia2025novel} formulated a Bayesian–Stackelberg game to model attacker-defender interactions under incomplete information, whereas Singh et al.~\cite{singh2022dynamic} formulated the attacker–defender interaction as a linear quadratic differential game to design defensive strategies under asymmetric information. Similarly}, \revision{the authors in~\cite{hammar2022intrusion, anwar2022honeypot} modeled networked systems or individual devices using Markov models to guide defensive decision-making. The belief state, updated via Bayes' rule, was then used to select appropriate defensive actions (e.g., blocking IP addresses, allocating honeypots).} 
\revision{In a related line of work,\tcnsrevision{~\cite{shen2026privacy, peng2025propagation} employed epidemiological models to determine mitigation or remediation strategies, respectively.} However, these approaches do not consider the combination of lightweight and heavyweight defenses; lightweight defenses can be effective in delaying the attacker's progression, whereas heavyweight defenses (e.g., eviction) are necessary for restoring system integrity.}

Our proposed solution is closely related to the zero trust architecture that considers hosts outside a corporate firewall malicious and requests verification to access a network~\cite{stafford2020zero}. In contrast to this work, our work considers that the attacker is already in the corporate network, and firewalls between subnetworks have default rules to enable regular communications, which leads to more effective blocking of subnetworks.  \vspace{-2mm}

\begin{figure}[!t]
\begin{center}
    \includegraphics[width = 0.8\linewidth]{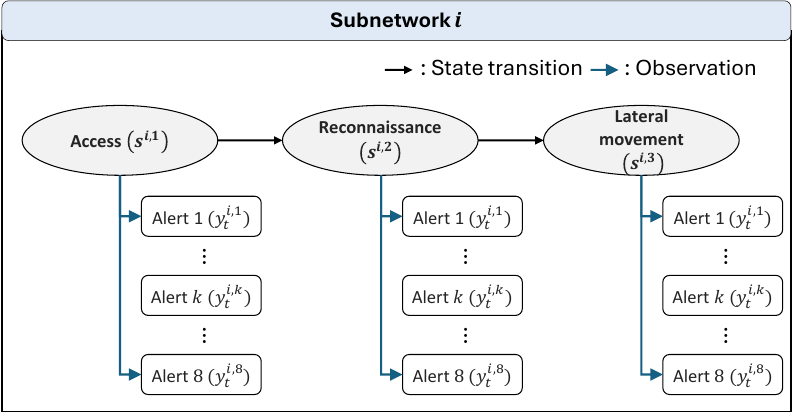}\vspace{-4mm}
      \caption{A subnetwork kill chain for $i \in \{1, \dots, I\}$. The row of gray nodes represents the attack states in subnetwork $i \in \iSet \setminus \iSet^c$. Black arrows denote state transitions, while blue arrows indicate observations conditioned on the state. The attacker repeatedly conducts lateral movements across transient networks in an effort to reach one of the critical subnetworks.}\label{fig:attack_model}
\end{center}
\vspace{-8mm}
\end{figure}

\section{System and Attack Model}
\label{sec:attack_model}

We consider a system that consists of a set $\iSet = \{1,\dots,I\}$ of subnetworks separated by firewalls, in accordance with state-of-the-art practice in industrial control systems~\cite{bicaku2020security}, as illustrated in Fig.~\ref{fig:first_page}. We denote by $\iSet^{\uparrow}(i)$ the set of upstream subnetworks relative to subnetwork $i$ (e.g., $\iSet^{\uparrow}(4)=\{1,2,3\}$), by $\iSet^{\uparrow}_d(i) \subseteq \iSet^{\uparrow}(i)$ the set of directly connected upstream subnetworks relative to subnetwork $i$ (e.g., $\iSet^{\uparrow}_d(4)=\{2,3\})$, by $\iSet^{\downarrow}(i)$ the set of downstream subnetworks relative to subnetwork $i$ (e.g., $\iSet^{\downarrow}(3)=\{4,5\}$), and by $\iSet^{\downarrow}_d(i) \subseteq \iSet^{\downarrow}(i)$ the set of directly connected downstream subnetworks from subnetwork $i$ (e.g., $\iSet(3)=\{4\}$). We denote the set of critical subnetworks by $\iSet^c \subset \iSet$;  subnetworks $i\in \iSet^c$ contain critical assets, making it imperative to secure these subnetworks.
We make the reasonable assumption that the firewalls between neighboring subnetworks filter traffic based on rules that can be updated in real-time. 
%


\begin{figure*}[!t]
\begin{center}
  \includegraphics[width=0.8\linewidth]{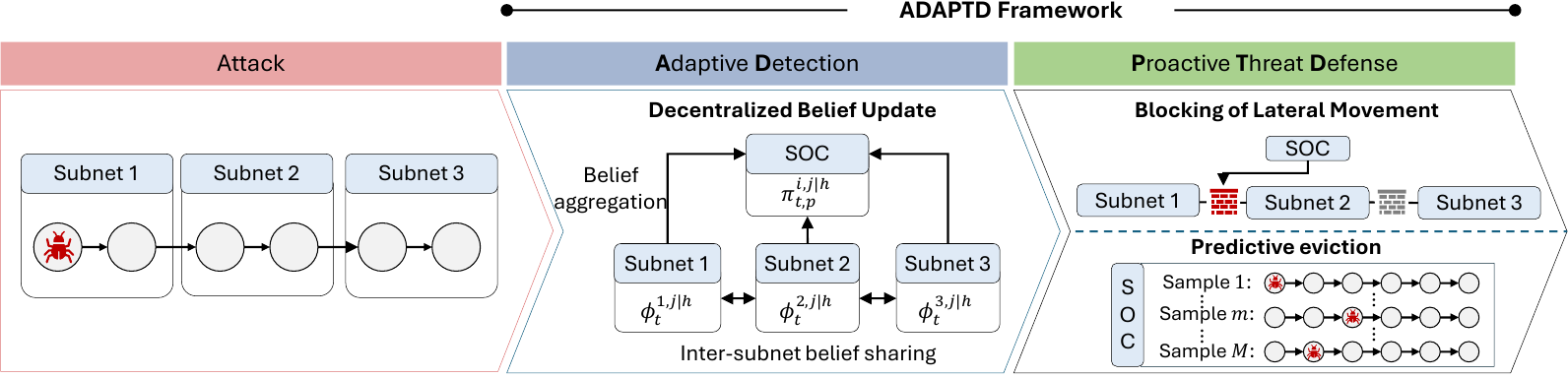}\vspace{-4mm}
  \caption{Illustration of the {\model} framework. The attacker initiates the intrusion from subnetwork 1. As the attack propagates, each subnetwork independently maintains its partitioned belief to track the attacker's progression, while the SOC aggregates these beliefs. Based on the aggregated belief, threat blocking contains the attacker, and predictive eviction removes the attacker from the system when the expected cost of continued blocking exceeds the predefined eviction cost.}
  \label{fig:framework}
\end{center}
\vspace{-8mm}
\end{figure*}

\vspace{-2mm}
\subsection{Attack model}\vspace{-2mm}

\subsubsection{Attack states} We model the attacker's progression within a subnetwork by transitions between attack states, where each attack state corresponds to a tactic in the MITRE ATT\&CK framework~\cite{mitre_attack}. 
Let us denote the set of attack state indices by $\jSet = \{1, \dots, j, \dots, J\}$ and the set of attack states in subnetwork $i$ by $\stageSet^i = \{\stage^{i,1}, \dots, \stage^{i,j}, \dots, \stage^{i,J}\}$. We adopt the convention that $\stage^{i,1}$ corresponds to access to subnetwork $i$ and $\stage^{i,J}$ corresponds to lateral movement from subnetwork $i$ to a downstream subnetwork. Since all subnetworks are initially clean, we denote this completely clean state by $\stage^{0,0}$. The overall state space is therefore $\stageSet=\cup_{i\in\mathcal{I}} \stageSet^i \cup \{\stage^{0,0}\}$. The objective of the attacker is to compromise assets within one of the critical subnetworks, which can be done by reaching state $s^{i,1}$ for some  $i \in \iSet^c$. 


Fig.~\ref{fig:attack_model} illustrates three attack states in a kill chain (i.e., $J=3$): initial access $\stage^{i,1}$, reconnaissance $\stage^{i,2}$, and lateral movement $\stage^{i,3}$.
The initial access state $\stage^{i,1}$ involves the use of various exploits, such as exploiting external remote services and public-facing applications, compromised removable media, and abused valid accounts.
The reconnaissance state $\stage^{i,2}$ includes malicious activities such as vulnerability and IP scanning, aiming at identifying vulnerable hosts in neighboring subnetworks. Finally, the lateral movement state $\stage^{i,3}$ uses tactics such as exploiting remote services, abusing the and hijacking remote sessions to compromise neighboring subnetworks. Notably, techniques such as remote‑service exploitation and the abuse of valid accounts can be used in both initial access and lateral movement. To compare these two phases, we analyze the techniques and their associated procedure examples in the ATT\&CK Enterprise matrix~\cite{mitre_attack}. The analysis shows that 58\% of the procedure examples associated with lateral movement also appear among the procedure examples for initial access. At the technique level, these overlapping procedures correspond to 89\% of all lateral‑movement techniques. This substantial overlap motivates our modeling choice: access to a downstream subnetwork $i^\prime \in \iSet^{\downarrow}_d(i)$ following lateral movement in subnetwork $i$ is modeled as an instance of initial access.
Given that the attacker's goal is to compromise a critical subnetwork, in  Fig.~\ref{fig:attack_model}, we omit minor attack steps unrelated to the lateral movement. This is consistent with the notion of a high-level kill chain~\cite{holgado2017real}, which demonstrated that some attack steps can be excluded while detecting attacks successfully.

\subsubsection{Attacker progression}
Time is considered to be slotted and is represented by $t \in \mathbb{Z}^{+}$, where $\mathbb{Z}^{+}$ denotes the set of positive integers. We denote the attack state at time $t$ by $\stageRv_t  \in \stageSet$, and we consider that $\stageRv_0=\stage^{0,0}$, i.e., all subnetworks are initially not compromised. The attacker then selects a subnetwork and initiates an attack. Let us denote by $\prior^i$ the probability that the attacker selects subnetwork $i$ and by $\P(\stageRv_{t+1} = \stage^{i,1} \mid \stageRv_{t} = \stage^{0,0})=1-\P(\stageRv_{t+1} = \stage^{0,0} \mid \stageRv_{t} = \stage^{0,0})$ that the probability of initiating attack at time $t$. We make the reasonable assumption that the initial compromise does not happen in a critical subnetwork, i.e., $\sum_{i \in \iSet \setminus \iSet^c} \prior^i = 1$, that is, $\prior^{i^{\prime}} = 0$ if $i^\prime \in \iSet^{c}$. This assumption is reasonable  considering the stringent security controls mandated by standards such as IEC62443~\cite{bicaku2020security}. 

We denote the probability that the attacker progresses from state $s^{i,j}$ to state $s^{i^\prime,j^\prime}$ between subsequent time slots by $\tranProb^{i,j \rightarrow i^\prime,j^{\prime}} = \P(\stageRv_{t+1} = \stage^{i^\prime,j^{\prime}} \mid \stageRv_{t} = \stage^{i,j})$. 
As we consider an APT attacker, we make the reasonable assumption that the attacker aims to compromise a critical subnetwork without re‑compromising subnetworks that have already been compromised. Therefore, the attacker tries to transition to a directly connected downstream subnetwork $i^\prime \in \iSet^{\downarrow}_d(i)$ from subnetwork $i$.
Recall that state $\stage^{i,J}$ corresponds to lateral movement, thus a transition from subnetwork $i \in \iSet$ to a direct downstream subnetwork $i^\prime \in \iSet^{\downarrow}_d(i)$ would correspond to a transition to $\stage^{i^\prime,1}$, i.e., initial access on subnetwork $i^\prime$. 
\vspace{-2mm}
\subsection{Defender model}
\label{sec:defender_model}
\subsubsection{Observations and Belief} The defender can observe alerts generated by the \ac{IDS}, but cannot directly observe the attacker's underlying state. Considering that the \ac{IDS} does not trigger true alerts when the attack is not present in subnetwork $i$, the defender extends the kill chain in Fig.~\ref{fig:attack_model} by prepending the clean state $\stage^{i,0}$ and appending the foothold state $\stage^{i,J+1}$, which represent an uncompromised subnetwork and inactive exploits after the attacker has left the subnetwork, respectively. Thus, we denote the defender's states by $\stageSet^\prime =\cup_{i\in\mathcal{I}} \stageSet^i \cup_{i\in \iSet}\{\stage^{i,0}, \stage^{i,J+1}\}$, the set of alert indices in a subnetwork by $\kSet = \{1,\dots,k,\dots,K\}$, and the alerts from subnetwork $i$ observed at time slot $t$ by $\alertVectorRv^{i}_t = [\alertVectorRv^{i,k}_t] \in \{0,1\}^K$, where $\alertVectorRv^{i,k}_t = 1$ indicates an alert, and $\alertVectorRv^{i,k}_t = 0$ indicates no alert. The alerts across all subnetworks at time $t$ are denoted by the matrix $\alertVectorRv^{\mathcal{I}}_t = [\alertVectorRv^{i}_t] \in\{0,1\}^{I \times K}$. Finally, the tensor of alerts available to the defender at time $t$ is denoted by $\alertVectorRv_{1:t} = [\alertVectorRv^{\mathcal{I}}_{t^\prime}] \in\{0,1\}^{t \times I\times K}$  for $t^\prime \in \{1,\dots,t\}$.

We denote by $\tpRate^{i^\prime,k \vert i,j}$ the probability of a true alert $k$ in subnetwork $i^\prime \in \iSet$ given state $\stage^{i,j}$. We denote the probability of a false alert $k$ in subnetwork $i^\prime$ by $\fpRate^{i^\prime,k}$, which is independent of the attack state. The probability of observing  alert $k$ in subnetwork $i^\prime$ given attack state $\stage^{i,j}$ at time $t$ can thus be expressed as\vspace{-2mm}
\begin{align}
\begin{split}
            \P(\alertVectorRv^{i^\prime,k}_t &= \alertVector^{i^\prime,k}_t \vert \stageRv_t = \stage^{i,j})\\ 
            &= \begin{cases}
                 (1-\tpRate^{i^\prime,k \vert i,j})(1-\fpRate^{i^\prime,k}), & \alertVector^{i^\prime,k}_t=0, \\
                  1- (1-\tpRate^{i^\prime,k \vert i,j})(1-\fpRate^{i^\prime,k}), & \alertVector^{i^\prime,k}_t=1.
            \end{cases}
\end{split}
\end{align}
Since the alerts are conditionally independent given the attack state, we can express the probability of observing a particular alert vector for subnetwork $i^{\prime} \in \iSet$ given a state at time $t$ as:
            \begin{align}
            \begin{split}
                 \P(\alertVectorRv^{i^{\prime}}_t &= \alertVector^{i^{\prime}}_t \vert \stageRv_t = \stage^{i,j}) = \prod_{k \in \kSet}\P(\alertVectorRv^{{i^{\prime}},k}_t = \alertVector^{i^{\prime},k}_t \vert \stageRv_t = \stage^{i,j}).
            \end{split}
            \end{align}
Then, the probability of observing alert vector $\alertVector^{\iSet}_t$ in state $\stage^{i,j}$ can be expressed as
        \begin{align}
           \begin{split}
                 \P(\alertVectorRv^{\iSet}_t = \alertVector^{\iSet}_t \vert \stageRv_t &= \stage^{i,j})= \prod_{i^{\prime} \in \iSet}\P(\alertVectorRv^{i^{\prime}}_t
                 = \alertVector^{i^{\prime}}_t \vert \stageRv_t = \stage^{i,j}).
            \end{split}
        \end{align}

%
\subsubsection{Actions} In every time slot, the defender can update the network \acp{ACL} (e.g., firewall rules), which it can use for blocking and unblocking traffic between subnetworks. We use $\dSet_t \in \powerSet(\iSet)$ to denote the set of subnetworks from which the defender blocks traffic towards their downstream subnetworks. That is, for a subnetwork $i$ and its direct downstream subnetwork $i^\prime \in \iSet^{\downarrow}_d(i)$, the transition probability becomes \vspace{-3mm}
\begin{align}
     \tranProb^{i,J \rightarrow i^{\prime},1}_{t} = \begin{cases}
         0,  \qquad & i \in \dSet_{t}, \\
        \tranProb^{i,J \rightarrow i^{\prime},1}, &  i \notin \dSet_{t}, 
    \end{cases} \vspace{-3mm}
\end{align}
i.e., the attacker cannot access subnetwork $i^{\prime}$ through lateral movement from subnetwork $i\in\dSet_t$. To limit the impact on business continuity, we consider that the defender's blocking decision is subject to the constraint $\vert \dSet_t \vert \leq B$ where $B\in \mathbb{N}$ is the blocking budget.

As an alternative, the defender can decide to escalate the incident to the emergency response team for immediate eviction of the attacker. We refer to this as the decision to evict, and we denote the decision at time $t$ by $\tau_t \in \{0,1\}$, where $\tau_t = 0$ means continuing regular operation, and $\tau_t=1$ means eviction of the attacker. Upon eviction, computers in all subnetworks are reinstalled to evict the attacker.

\subsubsection{Belief update}\label{sec:central_belief} Based on the observations and its actions, the defender can maintain a belief $\belief_{t+1}^{i, j}$ about the attacker's progression using the forward algorithm~\cite{fan2001constrained} as\vspace{-2mm}
\begin{align}
   \belief^{i,j}_{t+1} &= \mathcal{F}^{i,j}(\belief_t, \alertVector^{\iSet}_{t+1}, \dSet_{t}),\\
    &=\eta_t\cdot \alpha^{i,j}_{t+1} \cdot \P(\alertVectorRv^{\iSet}_{t+1} = \alertVector^{\iSet}_{t+1} \vert \stageRv_t = \stage^{i,j}),\label{eq:belief}\vspace{-2mm}
\end{align}
where $\mathcal{F}^{i,j}(\belief_t, \alertVector^{\iSet}_{t+1}, \dSet_t)$ is the forward algorithm using the previous belief $\belief_t$, current alerts $\alertVector^{\iSet}_{t+1}$ from all subnetworks, and the previous blocking decision $\dSet_t$. Each term in~\eqref{eq:belief} is\vspace{-1mm}
\begin{align}
    \begin{split}
    \alpha^{i,j}_{t+1} &= \sum_{i^{\prime} \in \iSet, j^\prime\in \jSet} \belief^{i^\prime,j^\prime}_{t} \cdot  \tranProb^{i^{\prime},j^\prime \rightarrow i,j}_{t}, \label{eq:alpha_central}
    \end{split}\\
    \begin{split}
       \eta_t &=\frac{1}{\sum_{i\in\iSet, j\in \jSet}\alpha^{i,j}_{t+1} \P(\alertVectorRv^{\iSet}_{t+1} = \alertVector^{\iSet}_{t+1} \vert \stageRv_t = \stage^{i,j})}.
    \end{split}
\end{align}
It is worth noting that the forward algorithm has to collect alerts from multiple subnetworks to compute belief. This can make it challenging to maintain the belief $\belief_t$ in real-world systems.
\subsubsection{Cost Model}
Our cost model captures four aspects: the cost of a compromise, the cost of network reconfiguration, its impact on business continuity, and the cost of false eviction. 
We denote the cost of the compromise of subnetwork $i$ by $c^i_c$. 
We consider that the act of reconfiguring the \ac{ACL} incurs cost $c_d$ per subnetwork. Furthermore, we model business continuity in terms of the value $r^{i,i^\prime}$ of connectivity between pairs of subnetworks $i$ and $i^\prime$, hence blocking traffic will result in corresponding loss of revenue.   
Finally, we denote the cost of false eviction by $c_\tau$, representing the cost of unnecessary reinstallation of system components even though they were not compromised.
\vspace{-3mm}
\subsection{Problem Formulation}
We consider that the defender is a rational entity. We can express the perceived cost of the defender at time $t$ as \vspace{-5mm}

\begin{align}\label{eq:perceivedCost}
\begin{split}
c(\beliefRv_t, \dSet^\policy_t, \tau_t)&=  \hspace{-2mm}\sum_{i \in \iSet}\sum_{i^\prime \in \iSet^{\downarrow}_d(i)} c^{i^\prime}_c   \beliefRv^{i,J}_t \tranProb^{i,J\rightarrow i^\prime,1} + \hspace{-1mm}   c_d \vert \dSet^{\policy}_{t} \setminus \dSet^\policy_{t-1} \vert \hspace{-1mm}  \\ 
&\quad -  \hspace{-3mm}\sum_{i\in \iSet \setminus \dSet^\policy_t}  \sum_{i^\prime \in \iSet^{\downarrow}_d(i)}\hspace{-1mm}   r^{i,i^\prime}+ \tau_t c_\tau 
(1\hspace{-1mm} -\hspace{-1mm}\beliefRv^{0,0}_t)
 , 
\end{split}
\end{align}
where $\policy$ is a policy for blocking and eviction. In~\eqref{eq:perceivedCost}, the first term is the risk of the attacker compromising another subnetwork at time $t+1$. The second term is the cost of reconfiguring the firewall. The third term is the value of connectivity between non-blocked subnetworks. 
Finally, the fourth term is the cost incurred by falsely evicting the attacker.


The defender's objective is to minimize the cost of a networked system, subject to a constraint on the blocking budget $B$. We denote by $\gamma \in [0,1]$ the discount factor, by $\policy_{\dSet} \in \policySet_{\dSet}$ the policy for blocking subnetworks, and by $\policy_{\tau} \in \policySet_{\tau}$ the policy for eviction. Then, we express the defender's objective as \vspace{-5mm}
\begin{align}
   \min_{\policy_{\dSet} \in \policySet_{\dSet}, \policy_{\tau} \in \policySet_{\tau}} \hspace{-1mm} &\E \left\{  \sum_{t=1}^{\infty} \gamma^{t} c(\beliefRv_t, \dSet^{\policy_{\dSet}}_t, \tau_t) \mathbbm{1}_{\sum_{t^\prime=1}^{t-1}\tau_{t^\prime}=0} \vert \beliefRv_0 =\belief_0 \hspace{-0.5mm}\right\} \hspace{-0.5mm}, \label{eq:final_objective}
\\
   s.t \quad & \beliefRv^{i,j}_{t+1} = \mathcal{F}^{i,j}(\beliefRv_t,\alertVectorRv^{\iSet}_{t},\dSet_t),\\
    \quad &\dSet^{\policy_{\dSet}}_t = \policy_{\dSet}(\beliefRv_t),\\
    \quad &\tau_t = \policy_{\tau}(\beliefRv_t),\\
    \quad & \vert \dSet^{\policy_{\dSet}}_t \vert \leq B,
\end{align}
where $\mathbbm{1}$ is an indicator function. The defender's problem is thus a \ac{POMDP} with an infinite horizon, whose length depends on the defender's policy.

A straightforward approach to solving the above problem would be to train a defender policy using deep reinforcement learning (RL) ~\cite{kaelbling1996reinforcement}. However, this approach would not provide explainability~\cite{bekkemoen2024explainable} and would require extensive training of $\vert \iSet \setminus \iSet^c \vert$ RL models due to the random initial compromise.
Instead, our objective is to develop a solution that (i) minimizes the defender's cost, and at the same time (ii) is  explainable and (iii) has low communication overhead. In what follows we present \emph{{\model}}, a framework and algorithms for achieving the above design objectives.

\begin{figure}[!t]
\begin{center}
    \includegraphics[width = 0.8\linewidth]{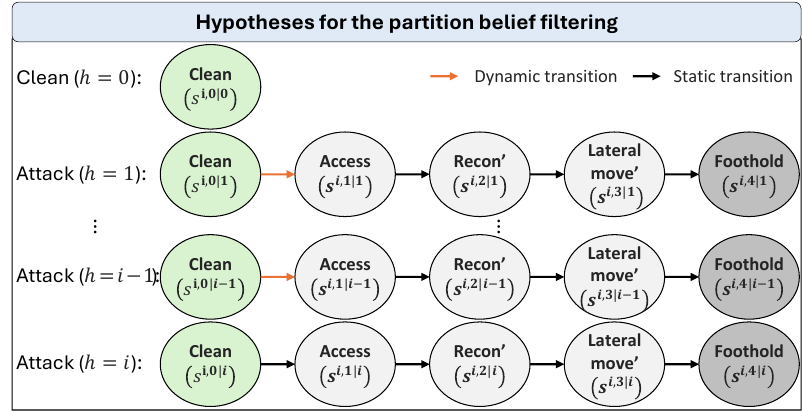}\vspace{-4mm}
    \caption{Hypotheses in subnetwork $i$ for the {\chain} HMM. Each oval represents a state, and each subnetwork contains a sequence of states. Hypothesis $h$ corresponds to an attacker initiating access from subnetwork $h$ and subsequently reaching subnetwork $i$.}
    \label{fig:chained_hypohtheses}
\end{center}
\vspace{-8mm}
\end{figure}

\section{{\model} Framework}
\label{sec:adapti}


Fig.~\ref{fig:framework} illustrates the interaction between the attacker and the defender in our proposed framework. The framework is designed to enable agile attack detection and to mitigate lateral movement and the compromise of critical assets. {\model} consists of three core modules:
\begin{itemize}
\item \textbf{Decentralized Belief Update}: This module enables to maintain real-time \ac{SA} at low overhead. The core idea is that each subnetwork maintains a belief about the attacker progression in the subnetwork based on the locally generated alerts and based on the beliefs of neighboring subnetworks. Decentralized belief update allows low communication overhead compared to sharing alerts across all subnetworks (see Sec.~\ref{sec:distributed_belief}).
\item \textbf{Blocking of Lateral Movement}: This module isolates subnetworks to block the attacker's progression. The decision to isolate subnetworks is taken based on the current belief, computing the cost reduction by blocking one-step progression of the attack and greedily selecting subnetworks to block (see Sec.~\ref{sec:threat_blocking}).
\item \textbf{Predictive Eviction}: This module estimates the expected blocking cost via \ac{MC} simulations and determines whether eviction would be beneficial. Eviction is triggered when the expected blocking cost exceeds the termination cost. (see Sec.~\ref{sec:proactive_termination}).
\end{itemize}
In what follows, we present the three modules in detail.
\subsection{Decentralized belief update}
\label{sec:distributed_belief}
The centralized belief update in Sec.~\ref{sec:central_belief} requires aggregating all raw alerts at a single node. Such a scheme is vulnerable to link failures and is constrained by bandwidth and energy resources~\cite{ghasemi2010stochastic, tamjidi2020efficient}. These limitations motivate the use of decentralized belief update, which mitigates both communication and reliability concerns. However, existing decentralized approaches typically rely on consensus among subnetworks, incurring additional communication rounds. To address this issue, we propose {\chain}, which employs dynamic transition probabilities and eliminates the consensus phase to reduce communication overhead.

We capture diverse attack vectors through a set of hypotheses. Let us denote by $\hSet =\{0, \dots, I\}$ the set of hypotheses representing the initiation of attack from subnetwork $h \in \hSet$ where $h=0$ corresponds to the clean scenario, by $\hSet^i = \iSet^{\uparrow}(i) \cup \{i\} \cup \{0\}$ the set of hypotheses that can compromise subnetwork $i$ (e.g., $\hSet^4 = \{0,1,2,3,4\}$), by $\iSet^h =\{h\} \cup \iSet^{\downarrow}(h)$ the set of reachable subnetworks starting from subnetwork $h$ (e.g., $\iSet^2 = \{2,4,5\}$ in Fig.~\ref{fig:first_page}), by $\stage^{i,j\vert h}$ the attack state $\stage^{i,j}$ under hypothesis $h \in \hSet$, $\tranProb^{i,j \rightarrow i^\prime,j^\prime \vert h }_t$ the dynamic transition probability under hypotehsis $h$ at time $t$, and by $\distributedBelief^{i,j\vert h}_{t}$ the partitioned belief for state $s^{i,j \vert h}$ under hypothesis $h \in \hSet$.  

For intuitive understanding of {\chain}, we consider a networked system comprised of two subnetworks where $\vert \iSet^{\downarrow}_d(1) \vert =1$, $i^\prime \in \iSet^{\downarrow}_d(2)$ and $\vert \iSet^{\uparrow}_d(2) \vert =1$.  The probability of the attacker transitioning to the next subnetwork is $\tranProb^{2,0 \rightarrow 2, 1}_t = 0$ (.resp $\tranProb^{2,0 \rightarrow 2, 1}_t = \tranProb^{1,J \rightarrow 2, 1}$) when the attacker's state satisfies $\stageRv_{t-1} \neq \stage^{1,J}$ (.resp $\stageRv_{t-1} = \stage^{1,J}$). However, because the defender cannot observe the attacker's true state, we instead rely on the partitioned belief in the upstream subnetwork,  $\distributedBelief^{1,J\vert h}_{t}$. This motivates modeling the transition probability from the clean state to the access state as a dynamic quantity that depends on the defender's belief about lateral movement in upstream subnetworks.  Each subnetwork therefore shares only its lateral‑movement belief, $\distributedBelief^{i,J \vert h}_{t}$, with its direct downstream subnetworks $i^\prime \in \iSet^{\downarrow}_d(i)$.

Fig.~\ref{fig:chained_hypohtheses} illustrates a clean hypothesis in subnetwork $i$ and $\vert \{i\} \cup \iSet^{\uparrow}(i) \vert$ attack hypotheses. Each kill chain follows the same structure as Fig.~\ref{fig:attack_model}, except that transitions from the clean state $\stage^{i,0 \vert h}$ to the access states $\stage^{i,1\vert h}$ are dynamically modified for hypotheses $h \neq i$. For $h = i$, no dynamic adjustment is required since the attacker initiates its threat directly from subnetwork $h$. 

For $h \neq i$, the attacker can transition into subnetwork $i$ from any of its direct upstream subnetworks $i^\prime \in \iSet^\uparrow_{d}(i) \cap \iSet^h$. Therefore, the dynamic transition probability has to aggregate the contributions from all such upstream subnetworks. For hypothesis $h$ and subnetwork $i$, a subnetwork-level path is a finite sequence \vspace{-3mm}$$p= (p_0, p_1, \dots, p_l, \dots p_L),$$
where $L$ is the length of the sequence. This sequence satisfies \vspace{-3mm}
$$p_0=h, \quad p_L =i,\quad  \tranProb^{p_l,J\rightarrow p_{l+1},1}_{t}>0.$$
The set of all such directed paths is denoted by \vspace{-3mm}
$$\mathcal{P}(h\rightarrow i)= \{p\vert p_0=h, p_{L} =i,\tranProb^{p_{l,J}\rightarrow p_{{l+1},1}}_{t}>0\}.$$
The probability of a path $p$ is \vspace{-2mm}
$$P(p) = \prod_{l=1}^{L-1} \tranProb^{p_{l,J}\rightarrow p_{l+1},1}_{t}.$$
For an upstream subnetwork $i^\prime \in \iSet^\uparrow_{d}(i) \cap \iSet^h$, the transition weight under hypothesis $h$ is defined as\vspace{-2mm}
$$
\omega^{i^\prime \vert i,h} = \frac{\sum_{p \in \mathcal{P}(h\rightarrow i)}\mathbbm{1}_{\{i^\prime \in p\}} P(p)}{\sum_{p\in \mathcal{P}(h\rightarrow i) }P(p)}.
$$
This quantity represents the probability that the attacker reaches subnetwork $i$ via upstream subnetwork $i^\prime$ under the given hypothesis $h$.

For a subnetwork $i$, combining upstream beliefs and
transition weights, the dynamic clean-to-access transition probability for subnetwork $i$ is
\begin{align}
    \tranProb^{i,0\rightarrow i,1 \vert h}_t \hspace{-2mm}=\hspace{-1mm} \begin{cases}\hspace{-0.5mm}
        \tranProb^{0,J \rightarrow i,1}  ,&\hspace{-3mm}\text{for } h=i,\\\hspace{-0.5mm}
        \sum_{i^\prime\in\iSet^h \cap \iSet^{\uparrow}_d(i)}\omega^{i^\prime \vert i, h} \distributedBelief^{i^\prime,J \vert  h}_{t-1} \tranProb^{i^\prime,J \rightarrow i,1},  &\hspace{-3mm}\text{otherwise},
    \end{cases}   
\end{align}
while, for $j^\prime \neq 0$ and $j \neq 1$, $\tranProb^{i,j^\prime\rightarrow i,j \vert h}_t = \tranProb^{i,j^\prime \rightarrow i,j}$. 
Then, we then model the attacker's transition as.
 \begin{align}
    \begin{split}
   \acute{\alpha}^{i,j\vert h}_{t} &=\begin{cases}
\distributedBelief^{i,0 \vert h}_{t-1} \cdot  \tranProb^{i,0 \rightarrow i,0 \vert h}_t, \\\qquad \qquad \qquad \qquad \qquad \qquad \text{ for }j=0,
   \\
       \distributedBelief^{i,0 \vert h}_{t-1} \cdot  \tranProb^{i,0 \rightarrow i,j \vert h}_t + \sum_{j^{\prime} \in \jSet}\left(\distributedBelief^{i,j^{\prime} \vert h}_{t-1} \cdot  \tranProb^{i,j^{\prime} \rightarrow i,j \vert h}_t\right)  , \\\qquad \qquad \qquad \qquad \qquad \qquad \text{otherwise,}
   \end{cases}\label{eq:chain_prior}
    \end{split}
\end{align}
which uses the prior belief and the corresponding transition probabilities. Next, we compute the likelihood using the probability of local alerts,\vspace{-3mm}
 \begin{align}
    \begin{split}
   \alpha^{i,j\vert h}_{t} &=\P(\alertVectorRv^{i}_t = \alertVector^{i}_t \vert \stageRv_t = \stage^{i,j})\acute{\alpha}^{i,j\vert h}_{t}.\label{eq:chain_observation}
    \end{split}
\end{align}
The belief in subnetwork $i$ can be obtained by normalization:\vspace{-3mm}
\begin{align}
    \distributedBelief^{i,j\vert h}_{t} = \eta^{i,h}_{t,p}\alpha^{i,j\vert h}_{t}, \label{eq:chain_normalize}
\end{align}
where $\eta^{i,h}_{t,p} = \frac{1}{\sum_{j\in\jSet^i}\alpha^{i,j\vert h}_{t,p}}$.
Then, we can combine the local beliefs and obtain the aggregated belief as,
\begin{align}
    {\belief}^{i,j\vert h}_{t,p} = {\eta}^h_{t,p}\distributedBelief^{i,j\vert h}_{t} \hspace{-3mm}\prod_{i^\prime \in \iSet^h \cap \iSet^{\uparrow}(i)}\hspace{-3mm}\distributedBelief^{i^\prime,J+1\vert h}_{t}\hspace{-3mm}\prod_{i^\prime \in \iSet^h \setminus \iSet^{\uparrow}(i)}\hspace{-3mm}\distributedBelief^{i^\prime,0\vert h}_{t} ,\label{eq:chain_central}
\end{align}
 where ${\eta}^h_{t,p}$ is a normalization factor for hypothesis $h$.
Now, let us analyze the similarity between the central and the aggregated partitioned beliefs using \ac{KL} divergence. To this end, we introduce two assumptions for the partitioned belief as follows:
\begin{assumption}[Lower-Bounded Belief]\label{assume:belief_lowerbound}
There exists $\epsilon > 0$ such that for all $i \in \iSet$ and $j \in \jSet$,
\[
\distributedBelief^{i,j\vert h}_{t-1,p} \geq \epsilon, \text{ and } \belief^{i,j \vert h}_{t-1} \geq \epsilon.
\]
\end{assumption}
This assumption implies that we use $\epsilon$ in place of $0$. This is a \revision{necessary choice to ensure numerical stability}, as the \ac{KL} divergence involves a logarithm, which is undefined at $0$.

\begin{assumption}[Bounded Transition Probability]\label{assume:Bounded_transition}
The transition probabilities satisfy:
\[
q^{i^\prime,j^\prime \rightarrow i,j} \in [q_{\min}, q_{\max}] \subset (0,1).
\]
\end{assumption}
Under these two assumptions, we can compute the upperbound of the \ac{KL} divergence between central and {\chain} beliefs as,
\begin{proposition}
    The KL divergence satisfies
\begin{align}
    D_{KL}(\belief^{\cdot,\cdot \vert h}_t \parallel \belief^{\cdot,\cdot \vert h }_{t,p})&\leq
    \log \Bigg(\frac{q_{\text{max}}}{\sum_{j^{\prime} \in \jSet^{i \vert h}} \epsilon^{\vert \iSet^h \vert +1}   q_{\text{min}}^{\vert \iSet^h \vert}}\Bigg).
\end{align}
\end{proposition}
\begin{proof}
    See appendix.
\end{proof}

\subsection{Threat blocking}
\label{sec:threat_blocking}
The computational cost of predictive blocking based on \ac{MC} simulation is significant~\cite{miehling2015optimal,miehling2018pomdp}, \revision{and poses a fundamental challenge for real-time decision-making in large-scale networks}. Instead, we propose a computationally simple approach to subnetwork blocking, based on sequential hypothesis testing. Since, in our problem, there are multiple hypotheses and corresponding beliefs (in different subnetworks), the challenge lies in identifying the hypothesis and the belief to be used.
The proposed approach selects the subnetwork to be used by comparing their likelihood ratios at time $t$ and then determines the most likely hypothesis within that subnetwork.   For subnetwork $i$, the likelihood ratio is defined as
\begin{align}
        R^{i,h}_{t} &= \frac{\sum^{J+1}_{j=0} {\alpha}^{i,j\vert h}_t}{ {\alpha}^{i,0\vert 0}_t}, \label{eq:likelihood_testing}
    \end{align}
where this term represents the likelihood ratio between the attack hypothesis and the clean hypothesis. We then denote the most likely subnetwork by $$\hat{i}=\argmax_{i \in \iSet}\max_{h\in \hSet}R^{i,h}_t,$$
and the most likely hypothesis within that subnetwork  by $$\hat{h}= \argmax_{h \in \hSet} R^{\hat{i},h}_t.$$ Accordingly, the belief associated with hypothesis $\hat{h}$ is used for threat blocking.

Let us denote by $\dSet^{\leq B} \in \{ \dSet^\prime \vert \dSet^\prime \in \powerSet(\iSet), \vert \dSet^\prime\vert \leq B\}$ the set of subnetworks that can be blocked under the blocking budget $B$, by $\dSet_{t^\prime} \in \dSet^{\leq B}$ the set of subnetworks blocked at time $t^\prime \in \mathbb{N}$, and by $\dSet_{1:t-1} = [\dSet_1, \dots, \dSet_{t^\prime} \dots,  \dSet_{t-1}]$ the history of blocking decisions up to time $t-1$. Once our framework chooses the belief $\belief^{\cdot,\cdot \vert \hat{h}}_t$, the defender evaluates the expected cost at time $t$ without blocking. The per-time-step cost without blocking is given by:
\begin{align}\label{eq:oneTimeCost_no_block}
\begin{split}
c(\beliefRv_t\hspace{-1mm}=\hspace{-1mm}\belief^{\cdot,\cdot \vert \hat{h}}_t,\hspace{-0.5mm} \dRv_{1:t-1}\hspace{-1mm}=\hspace{-1mm}\dSet_{1:t-1})\hspace{-0.5mm}&=\hspace{-1mm} \sum_{i\in \iSet} \hspace{-1mm} \sum_{i^\prime \in \iSet^{\downarrow}_d(i)}\hspace{-1mm} c^{i^\prime}_c   \belief^{i,J \vert \hat{h}}_t \tranProb^{i,J\rightarrow i^\prime,1} \hspace{-1mm}\\&\quad -    r^{i,i^\prime}.
\end{split}
\end{align}
 With blocking $\dSet_t \in \dSet^{\leq B}$, the cost is
\begin{align}\label{eq:oneTimeCost_block}
\begin{split}
c(\beliefRv_t\hspace{-1mm}=\hspace{-1mm}\belief^{\cdot,\cdot \vert \hat{h}}_t, \hspace{-0.5mm}\dRv_{1:t}\hspace{-1mm}&=\hspace{-1mm}\dSet_{1:t})\hspace{-1mm}= \hspace{-2mm}\sum_{i\in \iSet \setminus \dSet_{t}}  \hspace{-0.5mm}\sum_{i^\prime \in \iSet^{\downarrow}_d(i)}\hspace{-1mm} c^{i^\prime}_c   \belief^{i,J\vert \hat{h}}_t \tranProb^{i,J\rightarrow i^\prime,1} \hspace{-1mm}- \hspace{-1mm}   r^{i,i^\prime}\\
&\qquad \qquad  +c_d \vert \dSet_{t} \setminus \dSet_{t-1} \vert. 
\end{split}
\end{align}
We define the cost reduction achieved by blocking subnetworks $\dSet_{t}$ as:
\begin{align}
\begin{split}
 c_{red}(&\beliefRv_t=\belief^{\cdot,\cdot \vert \hat{h}}_t,\dRv_{t}\hspace{-1mm}=\hspace{-1mm}\dSet_{t}) 
 \\&= c(\beliefRv_t=\belief^{\cdot,\cdot \vert \hat{h}}_t, \dRv_{1:t-1}=\dSet_{1:t-1}) \\
  &\qquad\qquad\qquad - c(\beliefRv_t=\belief^{\cdot,\cdot \vert \hat{h}}_t, \dRv_{1:t}=\dSet_{1:t}),
 \\&= \sum_{i \in \dSet_{t}}\hspace{-1mm}\sum_{i^\prime \in \iSet^{\downarrow}_d(i)}\hspace{-2mm}\left( \hspace{-0.5mm}c^{i^\prime}_c   \belief^{i,J\vert \hat{h}}_t \tranProb^{i,J\rightarrow i^\prime,1} \hspace{-0.5mm}-\hspace{-0.5mm}r^{i,i^\prime}\hspace{-0.5mm} \right)\hspace{-1mm}-\hspace{-1mm}c_d \vert \dSet_{t}\hspace{-1mm} \setminus \hspace{-1mm}\dSet_{t-1} \vert.\\
\label{eq:cost_reduction}
\end{split}
\end{align}

The greedy blocking at time $t$ selects an action that maximizes the cost reduction, i.e., $\hat\dSet_{t} = \argmax_{\dSet \in \dSet^{\leq B}} c_{red}(\beliefRv_t=\belief^{\cdot,\cdot \vert \hat{h}}_t, \dRv_{t}=\dSet)$. We make the blocking decision as, 
\begin{align}
    \dSet_{t} = \begin{cases}
        \hat{\dSet}_{t}, \quad \text{if } c_{red}(\beliefRv_{t}=\belief^{\cdot,\cdot \vert \hat{h}}_t, \dRv_{t}=\hat{\dSet}_{t}) > 0,\\ \emptyset, \quad \text{ otherwise.}
    \end{cases}
\end{align}
This implies that previously blocked subnetworks may be restored if the current belief suggests that blocking is no longer beneficial, i.e., when $c_{red}(\beliefRv_{t+t^\prime}=\belief_{t+t^\prime},\dRv_{t+t^\prime}=\dSet_{t+t^\prime}) \leq 0$ for $t^\prime \in \mathbb{N}$. This restoration mechanism is essential, as belief estimates are subject to error due to false alerts, and incorrect blockings must be reversed to maintain operational continuity. While the greedy decision strategy may not yield globally optimal blocking, it substantially reduces computational complexity, \revision{enabling real-time response and ensuring computational tractability in large-scale networks.}

\begin{table*}[!t]
\begin{center}
\caption{Transition probabilities associated with the composite kill chain constructed from Fig.~\ref{fig:first_page} and Fig.~\ref{fig:attack_model}.}\label{tab:transition_mat}
\resizebox{0.7\textwidth}{!}{%
\begin{tabular}{@{}ccccccccccccccccc@{}}
\toprule
\multirow{4}{*}{\textbf{\begin{tabular}[c]{@{}c@{}}Next\\ subnetwork\\ ($i^\prime$)\end{tabular}}} & \multirow{4}{*}{\textbf{\begin{tabular}[c]{@{}c@{}}Next\\ state\\ ($j^\prime$)\end{tabular}}} & \multicolumn{15}{c}{\textbf{Current subnetwork   ($i$)}}                                                                                                                                         \\ \cmidrule(l){3-17} 
                                                                                                   &                                                                                               & \multicolumn{3}{c}{\textbf{1}}       & \multicolumn{3}{c}{\textbf{2}}       & \multicolumn{3}{c}{\textbf{3}}       & \multicolumn{3}{c}{\textbf{4}}       & \multicolumn{3}{c}{\textbf{5}}       \\ \cmidrule(l){3-17} 
                                                                                                   &                                                                                               & \multicolumn{15}{c}{\textbf{Current state ($j$)}}                                                                                                                                                \\ \cmidrule(l){3-17} 
                                                                                                   &                                                                                               & \textbf{1} & \textbf{2} & \textbf{3} & \textbf{1} & \textbf{2} & \textbf{3} & \textbf{1} & \textbf{2} & \textbf{3} & \textbf{1} & \textbf{2} & \textbf{3} & \textbf{1} & \textbf{2} & \textbf{3} \\ \midrule
\multirow{3}{*}{\textbf{1}}                                                                        & \textbf{1}                                                                                    & 0.90       & 0.10       & 0.00       & 0.00       & 0.00       & 0.00       & 0.00       & 0.00       & 0.00       & 0.00       & 0.00       & 0.00       & 0.00       & 0.00       & 0.00       \\
                                                                                                   & \textbf{2}                                                                                    & 0.00       & 0.90       & 0.10       & 0.00       & 0.00       & 0.00       & 0.00       & 0.00       & 0.00       & 0.00       & 0.00       & 0.00       & 0.00       & 0.00       & 0.00       \\
                                                                                                   & \textbf{3}                                                                                    & 0.00       & 0.00       & 0.95       & 0.25       & 0.00       & 0.00       & 0.25       & 0.00       & 0.00       & 0.00       & 0.00       & 0.00       & 0.00       & 0.00       & 0.00       \\
\multirow{3}{*}{\textbf{2}}                                                                        & \textbf{1}                                                                                    & 0.00       & 0.00       & 0.00       & 0.95       & 0.05       & 0.00       & 0.00       & 0.00       & 0.00       & 0.00       & 0.00       & 0.00       & 0.00       & 0.00       & 0.00       \\
                                                                                                   & \textbf{2}                                                                                    & 0.00       & 0.00       & 0.00       & 0.00       & 0.95       & 0.05       & 0.00       & 0.00       & 0.00       & 0.00       & 0.00       & 0.00       & 0.00       & 0.00       & 0.00       \\
                                                                                                   & \textbf{3}                                                                                    & 0.00       & 0.00       & 0.00       & 0.00       & 0.00       & 0.97       & 0.00       & 0.00       & 0.00       & 0.03       & 0.00       & 0.00       & 0.00       & 0.00       & 0.00       \\
\multirow{3}{*}{\textbf{3}}                                                                        & \textbf{1}                                                                                    & 0.00       & 0.00       & 0.00       & 0.00       & 0.00       & 0.00       & 0.95       & 0.05       & 0.00       & 0.00       & 0.00       & 0.00       & 0.00       & 0.00       & 0.00       \\
                                                                                                   & \textbf{2}                                                                                    & 0.00       & 0.00       & 0.00       & 0.00       & 0.00       & 0.00       & 0.00       & 0.95       & 0.05       & 0.00       & 0.00       & 0.00       & 0.00       & 0.00       & 0.00       \\
                                                                                                   & \textbf{3}                                                                                    & 0.00       & 0.00       & 0.00       & 0.00       & 0.00       & 0.00       & 0.00       & 0.00       & 0.97       & 0.03       & 0.00       & 0.00       & 0.00       & 0.00       & 0.00       \\
\multirow{3}{*}{\textbf{4}}                                                                        & \textbf{1}                                                                                    & 0.00       & 0.00       & 0.00       & 0.00       & 0.00       & 0.00       & 0.00       & 0.00       & 0.00       & 0.98       & 0.03       & 0.00       & 0.00       & 0.00       & 0.00       \\
                                                                                                   & \textbf{2}                                                                                    & 0.00       & 0.00       & 0.00       & 0.00       & 0.00       & 0.00       & 0.00       & 0.00       & 0.00       & 0.00       & 0.98       & 0.03       & 0.00       & 0.00       & 0.00       \\
                                                                                                   & \textbf{3}                                                                                    & 0.00       & 0.00       & 0.00       & 0.00       & 0.00       & 0.00       & 0.00       & 0.00       & 0.00       & 0.00       & 0.00       & 0.98       & 0.02       & 0.00       & 0.00       \\
\multirow{3}{*}{\textbf{5}}                                                                        & \textbf{1}                                                                                    & 0.00       & 0.00       & 0.00       & 0.00       & 0.00       & 0.00       & 0.00       & 0.00       & 0.00       & 0.00       & 0.00       & 0.00       & 0.98       & 0.02       & 0.00       \\
                                                                                                   & \textbf{2}                                                                                    & 0.00       & 0.00       & 0.00       & 0.00       & 0.00       & 0.00       & 0.00       & 0.00       & 0.00       & 0.00       & 0.00       & 0.00       & 0.00       & 0.98       & 0.02       \\
                                                                                                   & \textbf{3}                                                                                    & 0.00       & 0.00       & 0.00       & 0.00       & 0.00       & 0.00       & 0.00       & 0.00       & 0.00       & 0.00       & 0.00       & 0.00       & 0.00       & 0.00       & 1.00       \\ \bottomrule
\end{tabular}
}
\end{center}
\vspace{-8mm}
\end{table*}
\subsection{Predictive eviction}\label{sec:proactive_termination}
Deep RL-based approaches~\cite{hammar2022intrusion, kurt2018online} have to be retrained whenever the underlying Markov model is modified, and such training is costly. To keep computational cost manageable, we propose to employ \ac{MC} simulation to estimate the expected cost of blocking and that of eviction. To further reduce the computational cost, we perform  a \ac{MC} simulation only if the maximum likelihood ratio exceeds the MC-trigger threshold $\threshold_{M}$. That is, for $\hat{i}=\argmax_{i \in \iSet}\max_{h\in \hSet}R^{i,h}_t$, we trigger a \ac{MC} simulation if the maximum likelihood ratio satisfies 
\begin{align}
    \max_{h \in \hSet}R^{\hat{i},h}> \threshold_M. \label{eq:4_likelihood_ratio}
\end{align}
If this is the case, then we perform a \ac{MC} simulation as follows. We sample $M$ particles, each represented as a tuple ($h_m, i_m, j_m$) where $h_m$ is the hypothesis, $i_m$ is the subnetwork, and $j_m$ is the attack state of $m^{th}$ particle. 
To sample a hypothesis, we compute the posterior distribution over hypotheses. Since the number of hypotheses in subnetwork $\hat{i}$ may be smaller than $H$, the computation of the posterior distribution in subnetwork $\hat{i}$  requires the prior distribution of hypotheses from the perspective of subnetwork $\hat{i}$. Recall that $q^{h}=\P(\hypothesisRv=h)$ is the prior probability of hypothesis $h$. The prior distribution from the perspective of subnetwork $\hat{i}$ is then obtained as
\begin{align}
    \P^{\hat{i}}(\hypothesisRv=\hypothesis) =\begin{cases}
        1- \sum_{h^\prime \in\hSet^{\hat{i}}\setminus \{0\}}q^{h^\prime} \quad &\text{for } h=0,
        \\
        q^{h^\prime} & \text{for } h^\prime\in \hSet^{\hat{i}}\setminus \{0\},\\
        0&\text{otherwise}.
    \end{cases}
\end{align}
We can then express the posterior distribution as
\begin{align}
    \P^{\hat{i}}(H=h \vert \alertVectorRv^{\hat{i}}_{1:t} &=\alertVector^{\hat{i}}_{1:t}) = \frac{\P^{\hat{i}}(H=h, \alertVectorRv^i_{1:t} =\alertVector^i_{1:t})}{\sum_{h\in\hSet^i}\P^{\hat{i}}(H=h, \alertVectorRv^{\hat{i}}_{1:t} =\alertVector^{\hat{i}}_{1:t})},\\
    &= \frac{\P^{\hat{i}}(\alertVectorRv^{\hat{i}}_{1:t} =\alertVector^{\hat{i}}_{1:t} \vert \hypothesisRv=\hypothesis)\P^{\hat{i}}(\hypothesisRv=\hypothesis)}{\sum_{h\in\hSet^{\hat{i}}}\P^{\hat{i}}( \alertVectorRv^{\hat{i}}_{1:t} =\alertVector^{\hat{i}}_{1:t} \vert \hypothesisRv=\hypothesis) \P^{\hat{i}}(\hypothesisRv=\hypothesis)}.
\end{align}
Once hypothesis $h_m$ is sampled from the posterior distribution (i.e., $h_m \sim \P^{\hat{i}}(H=h \vert \alertVectorRv^{\hat{i}}_{1:t} =\alertVector^{\hat{i}}_{1:t})$), we subsequently sample $i_m$ and $j_m$ from the belief $\belief^{i,j \vert h_m}_t$.

As we perform \ac{MC} simulation under blocking actions, using a large simulation horizon introduces bias. Once an attack hypothesis is selected, the simulation may result in continuous blocking, which artificially increases the expected blocking cost. To avoid this unfair effect, we restrict the simulation horizon to the expected time required for the attacker to reach the next subnetwork from the sampled state $s^{i_m,j_m}$. \revision{This horizon corresponds to the expected hitting time in Markov chains~\cite{norris1998markov}.}
Given that each state in our attack graph has two outgoing transitions (i.e., one to itself and one to the next state) the expected time to reach the next subnetwork is estimated by recursively enumerating the expected time to reach the next state. For a given state $s^{i,j}$, we define the set of directly reachable next states as $\mathcal{S}^{\downarrow}_d(i,j) = \{(i^\prime,j^\prime)\vert \tranProb^{i,j \rightarrow i^\prime,j^\prime}>0, s^{i,j}\neq s^{i^{\prime},j^\prime}\}$. The expected time to reach the next state is then computed as
\begin{align}
    \bar{t}^{i,j} = \frac{1}{\sum_{ (i^{\prime\prime},j^{\prime\prime})\in\mathcal{S}^{\downarrow}(i,j)}\tranProb^{i,j \rightarrow i^{\prime\prime},j^{\prime\prime}}}.
\end{align}
Thus, the expected time to reach the next subnetwork from state $\stage^{i,j}$ can be computed as
\begin{align}
    \bar{T}^{i,j} = \begin{cases}
        \bar{t}^{i,J }, &\text{if } j=J, \\
        \bar{t}^{i,J } + \sum^{J-1}_{j^\prime =j} \bar{t}^{i,j}, &\text{otherwise}.
    \end{cases}
\end{align}
Then, the cost of the $m^{th}$ \ac{MC} simulation is computed as 
\begin{align}
    c_m= \sum^{\bar{T}^{i_m,j_m}}_{t_m=0}\gamma^{t_m}\sum_{i \in \dSet_{t_m}}\sum_{i^\prime \in \iSet^{\downarrow}_d(i)}\left( r^{i,i^\prime} \right)+c_d \vert \dSet_{t_m} \setminus \dSet_{t_m-1} \vert.
\end{align}
We choose eviction if $\frac{1}{M}\sum^M_{m=1}c_m > c_\tau$, as it indicates that immediate eviction is more cost-efficient than continued operation given the current belief.

\subsection{\revision{Model Parameterization and Feasibility}}
{\model} requires an attack graph to model a networked system and probabilities to describe the attacker's behaviors The attack graph can be built based on  known exploits since finding zero-day vulnerabilities is difficult and developing exploits is time-consuming~\cite{ablon2017zero}. For example, one can use MulVAL for creating the attack graph~\cite{ou2005mulval}, and given the attack graph, the transition and observation probabilities can be obtained using automated penetration testing methodologies~\cite{holm2022lore,holgado2017real}.

\section{Numerical Results}
\label{sec:simulation}
In what follows, we evaluate {\model} using simulations on a composite kill chain that integrates Fig.~\ref{fig:first_page} and Fig.~\ref{fig:attack_model}.

\begin{table}[!b]
\vspace{-8mm}
\begin{center}
\caption{True and false alert rates used in the evaluation.}\label{tab:alert_prob}
\resizebox{0.6\columnwidth}{!}{%
\begin{tabular}{@{}ccccccccc@{}}
\toprule
\multicolumn{9}{c}{\textbf{False   alert rate}}                                                                                             \\ \midrule
\textbf{Subnet} & \multicolumn{8}{c}{\textbf{Alert index ($k$)}}                                                                          \\ \cmidrule(l){2-9} 
\textbf{Index ($i$)}  & \textbf{1}                 & \textbf{2} & \textbf{3} & \textbf{4} & \textbf{5} & \textbf{6} & \textbf{7} & \textbf{8} \\ \midrule
\textbf{1}         & 0.6                        & 0.4        & 0.5        & 0.4        & 0.5        & 0.4        & 0.5        & 0.4        \\
\textbf{2}          & 0.5                        & 0.3        & 0.4        & 0.3        & 0.4        & 0.3        & 0.4        & 0.3        \\
\textbf{3}          & 0.5                        & 0.3        & 0.4        & 0.3        & 0.4        & 0.3        & 0.4        & 0.3        \\
\textbf{4}          & 0.4                        & 0.2        & 0.3        & 0.2        & 0.3        & 0.2        & 0.3        & 0.2        \\
\textbf{5}          & 0.3                        & 0.1        & 0.2        & 0.1        & 0.2        & 0.1        & 0.2        & 0.1        \\ \midrule
\multicolumn{9}{c}{\textbf{Subnetwork $i \in \{1,...,4\}$}}                                                                                                \\ \midrule
\multicolumn{9}{c}{\textbf{True alert rate}}                                                                                                \\ \midrule
\textbf{Stage}      & \multicolumn{8}{c}{\textbf{Alert index ($k$)}}                                                                           \\ \cmidrule(l){2-9} 
\textbf{index ($j$)}  & \textbf{1}                 & \textbf{2} & \textbf{3} & \textbf{4} & \textbf{5} & \textbf{6} & \textbf{7} & \textbf{8} \\ \midrule
\textbf{1}          & 0.4                        & 0          & 0.3        & 0.1        & 0          & 0          & 0          & 0          \\
\textbf{2}          & 0                          & 0.3       & 0.2        & 0.2        & 0          & 0.2        & 0.2        & 0          \\
\textbf{3}          & 0                          & 0          & 0          & 0          & 0.3        & 0.3        & 0.3        & 0.4        \\ \midrule
\multicolumn{9}{c}{\textbf{Subnetwork $5$}}                                                                                                \\ \midrule
\multicolumn{9}{c}{\textbf{True alert rate}}                                                                                                \\ \midrule
\textbf{Stage}      & \multicolumn{8}{c}{\textbf{Alert index ($k$)}}                                                                           \\ \cmidrule(l){2-9} 
\textbf{index ($j$)}  & \textbf{1}                 & \textbf{2} & \textbf{3} & \textbf{4} & \textbf{5} & \textbf{6} & \textbf{7} & \textbf{8} \\ \midrule
\textbf{1}          & 0.4                        & 0          & 0.3        & 0.2        & 0          & 0          & 0          & 0          \\
\textbf{2}          & 0                          & 0       & 0        & 0        & 0.6          & 0.5        & 0        & 0          \\
\textbf{3}          & 0                          & 0          & 0          & 0          & 0       & 0        & 0.7        & 0.6        \\ \bottomrule
\end{tabular}}
\end{center}
\end{table}

\subsection{Evaluation Methodology}

\begin{figure*}[ht]
  \begin{minipage}[b]{0.5\linewidth}
  \centering
    \includegraphics[width = 0.75\linewidth]{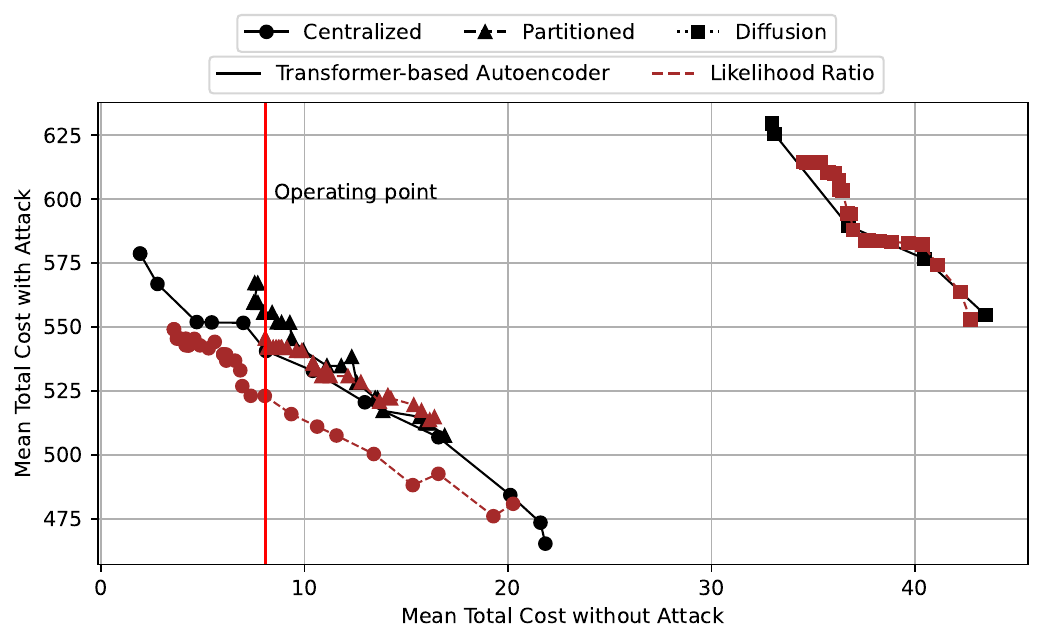}\vspace{-5mm}
    \caption[]
      {Mean total cost with attack vs. mean total cost without attack.}\label{fig:5_operating_point}
  \end{minipage}\hfill
  \begin{minipage}[b]{0.49\linewidth}
  \centering
    \includegraphics[width = 0.75\linewidth]{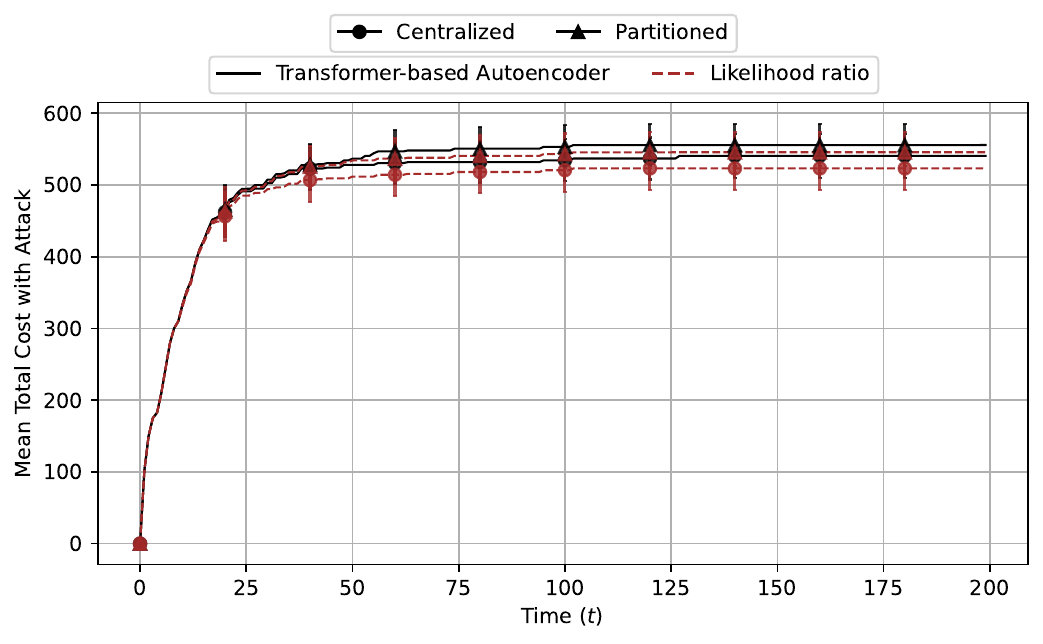}\vspace{-5mm}
    \caption{Mean total cost under attack vs. time including 95\% confidence intervals, for different belief update schemes and \ac{MC} triggering methods.}
\label{fig:5_time_dynamics}
  \end{minipage}
    \vspace{-5mm}
\end{figure*}

For the evaluation, we employ a composite kill chain constructed from Fig.~\ref{fig:first_page} and Fig.~\ref{fig:attack_model}.  The attack originates from any subnetwork $i \in \iSet \setminus \iSet^c$ excluding the subnetworks containing critical assets. We assume a uniform prior distribution over the choice of the initial subnetwork, given by $\prior^i=\frac{1}{\vert \iSet \setminus \iSet^c \vert}$. Once the initial subnetwork is selected, the probability of initiating malicious access is set to $\tranProb^{0,0 \rightarrow i,1}= 0.9$, reflecting a random initiation of the malicious access into the system. Table~\ref{tab:transition_mat} shows the transition probabilities of the full networked system. Once the initial subnetwork is accessed, the attacker proceeds according to the defined transition probabilities.  Table~\ref{tab:alert_prob} presents the true and false alert probabilities for each state $\stage^{i,j}$, with a higher false alert rates assumed for DMZ and corporate subnetworks due to their higher traffic volumes. The discount factor is $\gamma=0.97$. The operational reward for maintaining connectivity between subnetworks $i$ and $i^\prime$ is $r^{i,i^\prime} = \frac{2}{\vert \iSet^{\downarrow}_d(i) \vert }$. The cost of reconfiguring \ac{ACL} is $c_d=1$, and the cost of compromise for the $5$ subnetworks is  $[c^1_c, c^2_c,c^3_c,c^4_c, c^5_c]= [250, 500, 500, 750, 1500]$. Each simulation lasts $1,000$ time slots, with each time slot corresponding to a wall clock time  of five minutes.

For belief update, we use \emph{Partitioned} as a shorthand for {\chain}. As baselines for comparison, we consider two belief-update algorithms: \emph{Centralized} and \emph{Diffusion}~\cite{kayaalp2022hidden}. Using the first approach, the SOC aggregates alerts from all subnetworks and maintains a global belief state. While this method offers comprehensive threat visibility, it incurs substantial communication overhead. We refer to this as \emph{Centralized}. Using the second baseline approach, each subnetwork independently computes its belief from the full transition matrix based on local alerts, and then shares its complete belief with neighboring subnetworks for aggregation (see Diffusion HMM in Appendix). We refer to this as \emph{Diffusion}.

For predictive eviction, we employ the likelihood‑ratio test in~\eqref{eq:4_likelihood_ratio} and we used grid search to find the threshold $\threshold_M$. As a baseline for comparison, we use a transformer‑based autoencoder trained on sequences of benign likelihoods, and trigger a \ac{MC} simulations when the reconstruction error exceeds a threshold (details are provided in the Appendix).
The results shown are the averages of $200$ simulation runs. 

As the main evaluation metric, we use the incurred average total cost, computed akin to~\eqref{eq:perceivedCost}, but based on the actual progression of the attacker. 
As additional evaluation metrics we use the eviction delay, the false-eviction rate, and the empirical frequency of single blocking events. Let us denote by $t^{0,0 \rightarrow i,1}$ the starting time of the attack, i.e., the time when the attacker transitions from a clean state to access state and $t_{e}= \min\{ t \vert \tau_t=1, t\in \mathbb{N}\}$ the time of evicting the attacker. We then define the eviction delay \vspace{-2mm}
\begin{align}
    d_{e} = \max(t_{e} - t^{0,0 \rightarrow i,1},0).
\end{align}
We then define the false‑eviction rate as the fraction of simulation runs that incorrectly choose eviction when the scenario is not under attack. Finally, we define the empirical frequency of single-blocking events \vspace{-2mm}
\begin{align}
f^{single}= \frac{\sum_{t=1}^{T}\mathbbm{1}_{\vert \dSet_t \vert =1} }{\sum_{t=1}^{T}\mathbbm{1}_{\vert \dSet_t \vert>0}},
\end{align}
which captures how often a single blocking action is taken despite a potentially larger blocking budget.

For the computations we used a server with an Intel(R) Xeon(R) CPU E5-2620 v4 @ 2.10GHz and 32G memory. All algorithms were implemented in Python with Numpy and Pandas.

\subsection{Mean cost with attack vs. mean cost without attack}
We start with exploring the relationship between the mean total cost with attack and the mean total cost without attack.

Fig.~\ref{fig:5_operating_point} shows the mean total cost with attack as a function of the mean total cost without attack. The mean total cost with attack decreases as the mean total cost without attack increases, since a higher total cost without attack corresponds to lower \ac{MC}-trigger thresholds that result in higher false-eviction rates and faster true evictions. Regardless of the trigger method employed, \emph{Partitioned} exhibits performance comparable to \emph{Centralized} belief aggregation. Recalling that \emph{Centralized} incurs substantial communication overhead~\cite{ghasemi2010stochastic, tamjidi2020efficient}, we can conclude that \emph{Partitioned} is well suited for belief update in the considered setting. Since \emph{Partitioned} substantially outperforms \emph{Diffusion}, we omit \emph{Diffusion} from further analysis. The figure also shows that the transformer-based autoencoder achieves results similar to the likelihood ratio trigger. We will compare their performance later in more detail.  For the subsequent analysis, we choose an operating point where the mean total cost without attack is 8.07,  i.e., the minimum achievable cost for \emph{Partitioned}.

\vspace{-3mm}
\subsection{Mean total cost with attack vs. time}

Fig.~\ref{fig:5_time_dynamics} shows the mean total cost as a function of time, for  a mean total cost without attack of 8.07. The figure confirms that likelihood ratio‑based triggering of \ac{MC} simulations outperforms the transformer‑based autoencoder, indicating that the likelihood ratio offers an effective, training‑free \ac{MC} triggering mechanism, in contrast to a data‑driven approach that requires extensive training. Comparing the results for the different belief update schemes, we observe that the \emph{Partitioned} scheme achieves a cost comparable to the \emph{Centralized} scheme, i.e.,  similar performance at substantially reduced communication overhead, especially towards the end of the time horizon. At the beginning of an episode, however,  we observe a sharp cost increase due to the $0.1$ transition probability from clean state to initial access, regardless of \ac{MC}-triggering method and the belief update scheme, yielding an $88\%$ chance of transition by $t=20$ (i.e., $1-(0.9)^{20}$). Given the randomness and unpredictability of the initial compromise, the sharp increase is unavoidable. However, after about $t=25$, the performance of the different methods becomes clearly different, underscoring the importance of likelihood-ratio triggering for effective \ac{MC} simulation activation.

\begin{figure}[!t]
\begin{center}
    \includegraphics[width = 0.75\linewidth]{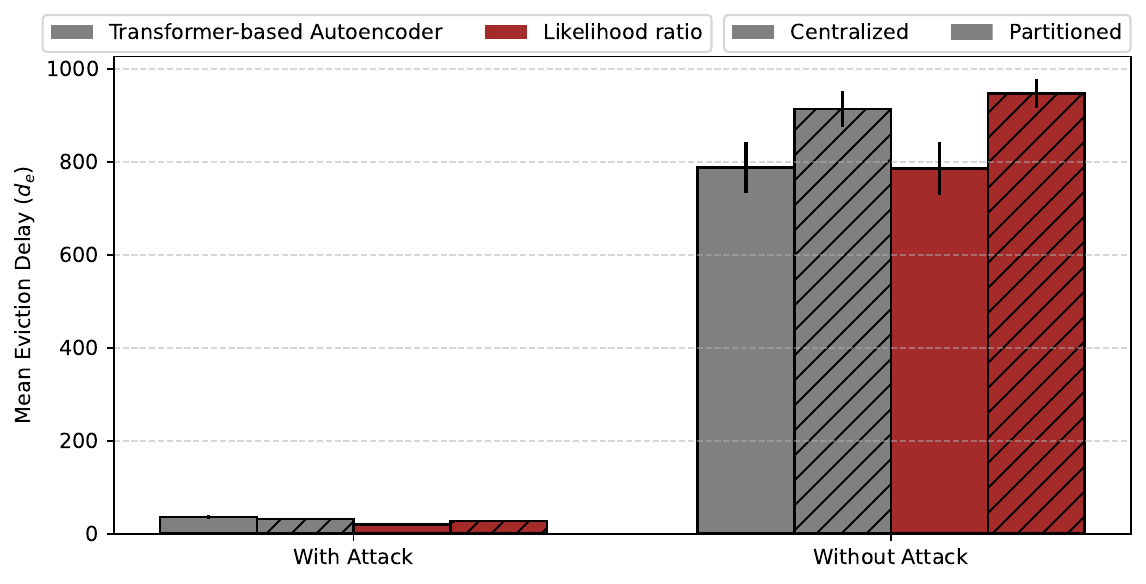}\vspace{-5mm}
    \caption[]
      {Mean eviction delay ($d_{e}$) with 95\% confidence intervals. A comparative analysis between two scenarios (i.e., one with attack and one without attack).}\label{fig:5_mtt}
\end{center}
\vspace{-8mm}
\end{figure}
\vspace{-5mm}
\subsection{Mean time to evict}
Fig.~\ref{fig:5_mtt} shows the mean time to eviction under different belief-update algorithms and triggering methods. We can observe the results are significantly different in scenarios with and without an attack, indicating that our framework responds appropriately to genuine threats while minimizing unnecessary evictions. Without attack, the maximum mean time to eviction reaches $1,000$ time steps, corresponding to the simulation horizon. Comparing belief-update algorithms, \emph{Partitioned} exhibits a mean time to eviction comparable to \emph{Centralized}, further confirming the superior scalability of \emph{Partitioned} without a significant degradation in eviction delay. Notably, \emph{Partitioned} with likelihood ratio demonstrates longer mean eviction delay without attack compared to the transformer-based autoencoder, indicating its suitability as \revision{a robust decision-theoretic engine that effectively minimizes false-eviction cost while maintaining threat responsiveness}. To analyze the differences introduced by the two \ac{MC} triggering methods, we next examine the frequency of MC simulations and the false‑eviction rate.

\begin{table}[b!]\vspace{-5mm}
\begin{center}
\caption{Average number of \ac{MC} simulations with 95\% confidence intervals and the corresponding false‑eviction rate.}\label{tab:5_mc_count_and_false_eviction}
\resizebox{\linewidth}{!}{%
\begin{tabular}{@{}ccccc@{}}
\toprule
                     & \multicolumn{2}{c}{\textbf{Count of MC   simulations}} & \multicolumn{2}{c}{\textbf{False-eviction rate}} \\ \midrule
                     & \textbf{Transformer}    & \textbf{Likelihood ratio}    & \textbf{Transformer}   & \textbf{Likelihood ratio}  \\ \midrule
\textbf{Centralized} & 2.58±1.06               & 0.63±0.22                    & 22.0\%                 & 22.0\%                     \\
\textbf{Partitioned} & 1.39±0.85               & 0.17±0.09                    & 9.0\%                  & 5.5\%                      \\ \bottomrule
\end{tabular}}
\end{center}
\end{table}

Table~\ref{tab:5_mc_count_and_false_eviction} shows the frequency of performing \ac{MC} simulations and the false-eviction rate without attack. 
The figure shows that \emph{Partitioned} with likelihood ratio results in substantially less \ac{MC} simulations compared to \emph{Centralized}, i.e., lower computational burden, as well in a lower false-eviction rate. The reason why \emph{Partitioned} results in less \ac{MC} simulations is that it considers a reduced hypothesis set ($|\hSet^{\hat{i}}| \leq H$ where $\hat{i}=\argmax_{i^\prime \in \iSet}\max_{h\in \hSet}R^{i^\prime,h}_t$), which in turn reduces the probability of falsely triggering an \ac{MC} simulation due to noisy alerts. Overall, the proposed likelihood ratio-based approach demonstrates superior performance compared to the transformer-based approach across configurations, confirming the effectiveness of our proposed framework. 

\vspace{-2mm}
\subsection{Impact of Blocking Budget}
\begin{figure*}[ht]
  \begin{minipage}[b]{0.49\linewidth}
  \centering
    \includegraphics[width = 0.75\linewidth]{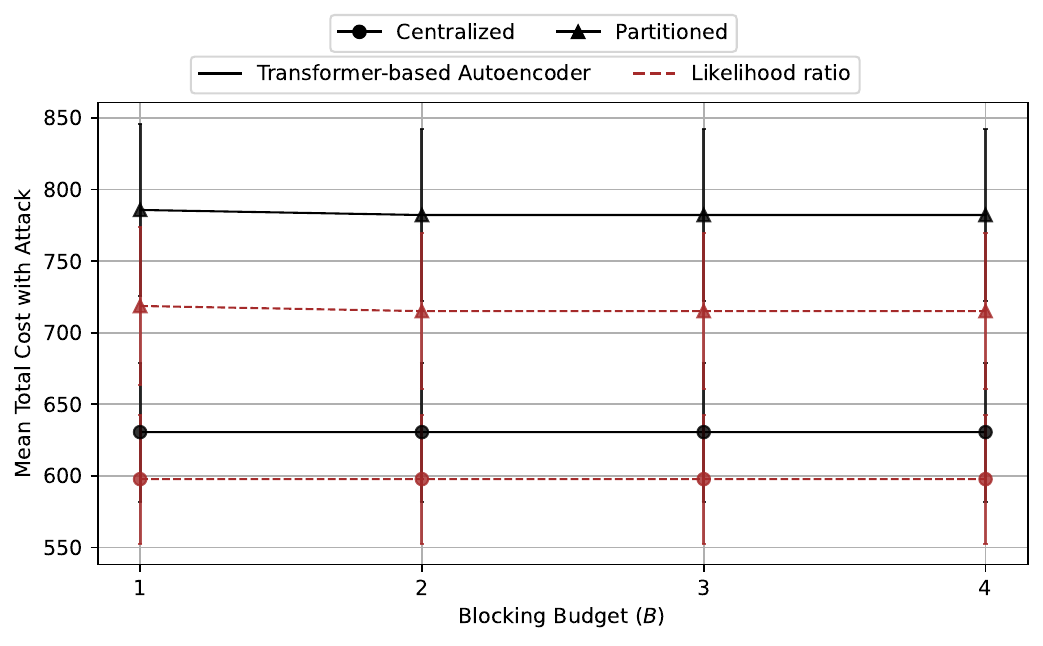}\vspace{-5mm}
    \caption
      {Mean total cost as a function of the blocking budget with 95\% confidence intervals.}\label{fig:5_budget}
  \end{minipage}\hfill
  \begin{minipage}[b]{0.49\linewidth}
  \centering
    \includegraphics[width = 0.75\linewidth]{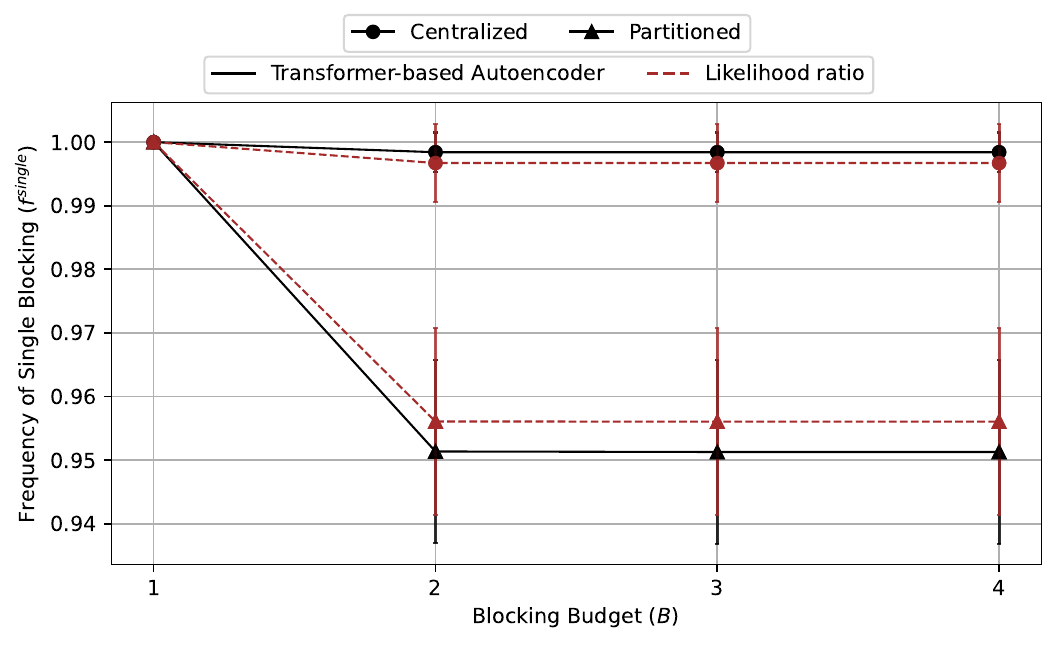}
      \caption{Frequency of single blocking events ($f^{single}$) as a function of the blocking budget, with 95\% confidence intervals.}\label{fig:5_fraction}
  \end{minipage}\vspace{-5mm}
\end{figure*} 
In what follows, we examine the impact of the blocking budget. When the belief is accurate, a single blocking action suffices. To introduce more uncertainty into the belief, we increase the false positive rate by 0.3 and rerun the grid search to compare the resulting total costs under attack with those obtained without attack at the same mean total cost. It is important to note that this scenario is extremely noisy, as all false alert rates become comparable to or even exceed the highest true alert rate.

Fig.~\ref{fig:5_budget} shows the mean total cost with attack as a function of blocking budget. The likelihood ratio consistently achieves lower total costs compared to the transformer-based autoencoder. Considering three aspects (i.e., the lower cost, the inherent explainability of the likelihood ratio, and the transformer's substantial training requirements), it confirms the likelihood ratio as the proper choice for our framework. Interestingly, \emph{Partitioned} exhibits a small cost reduction with increasing budget, while \emph{Centralized} remains largely unchanged; this stems from \emph{Partitioned}'s reliance on a smaller alert set for belief maintenance, which introduces belief inaccuracies absent in the more comprehensive \emph{Centralized} scheme. However, given that the increase in total cost from lateral movement ranges from $500$ to $1,500$, the additional cost incurred when using \emph{Partitioned} with the likelihood ratio is small, indicating that only a small fraction of simulations successfully achieve lateral movement. This highlights that our framework remains robust in defending the networked system even in a noisy environment. To further investigate the small difference in total cost, we examine the frequency of single blocking actions.

Fig.~\ref{fig:5_fraction} shows the frequency of single blocking actions as a function of the blocking budget. The figure shows that single blocking occurs in at least 95\% of cases across all belief-update algorithms and \ac{MC}-triggering methods, demonstrating the greedy algorithm's efficiency. Compared to \emph{Centralized}, \emph{Partitioned} shows a slightly lower single blocking frequency, reflecting belief uncertainty that prompts more extensive blocking; however, this blocking delays attacker progression and refines subsequent beliefs, ultimately making single blocking sufficient. In contrast, \emph{Centralized} maintains a higher single blocking frequency due to its more precise beliefs, albeit at greater communication cost. The cost differences observed in Fig.~\ref{fig:5_budget} primarily arise from single blocking dominance, underscoring the greedy blocking mechanism's flexibility in subnetwork selection based on belief states.

\vspace{-3mm}
\subsection{Ablation Study}
\label{sec:ablation}
Finally, we investigate the importance of the components in our proposed framework. 

Table~\ref{tab:5_ablation} shows the mean total cost with attack for different combinations of components. ADAPTD (blocking + termination) outperforms \emph{Eviction}, indicating the importance of \revision{integrating greedy containment and predictive recovery actions}. With ADAPTD (blocking + termination), \emph{Partitioned} shows performance comparable to \emph{Centralized}, demonstrating the \revision{scalability and stability} of our proposed decentralized solution. 
Without blocking, there is a large performance gap between \emph{Centralized} and \emph{Partitioned}, as blocking particularly benefits the decentralized approach: likelihood ratios computed locally in each subnetwork require complete re-accumulation upon attacker movement to new subnetworks. In contrast, \emph{Centralized} maintains cumulative likelihood ratios irrespective of attacker progression, enabling faster detection but incurring substantial communication overhead.


\begin{table}[b!]
\begin{center}\vspace{-5mm}
\caption{Ablation study. Mean cost with attack with 95\% confidence interval.}\label{tab:5_ablation}
\resizebox{0.8\linewidth}{!}{%
\begin{tabular}{@{}cccc@{}}
\toprule
\multicolumn{4}{c}{\textbf{Total cost with Attack}}                                                                                                                                          \\ \midrule
\textbf{Method}                                                                                        & \textbf{Belief algorithm} & \textbf{Eviction} & \textbf{Blocking + Eviction} \\ \midrule
\multirow{2}{*}{\textbf{\begin{tabular}[c]{@{}c@{}}Transformer-based\\      Autoencoder\end{tabular}}} & \textbf{Centralized} & 560.60±31.30         & 540.66±30.73                    \\
                                                                                                       & \textbf{Partitioned} & 628.75±36.26         & 555.69±29.84                    \\ \midrule
\multirow{2}{*}{\textbf{\begin{tabular}[c]{@{}c@{}}Likelihood\\      Ratio\end{tabular}}}              & \textbf{Centralized} & 534.80±32.20         & 523.12±29.74                    \\
                                                                                                       & \textbf{Partitioned} & 596.40±34.98         & 545.60±28.66                    \\ \bottomrule
\end{tabular}}
\end{center}
\end{table}

\vspace{-3mm}
\section{Conclusion}
\label{sec:conclusion}
We addressed the challenge of threat response against autonomous APTs by formulating the problem as a \ac{POMDP}. The key innovations of our framework include: i) the flexible detection of threats, ii) the immediate blocking for containment, and iii) proactive threat defense to prevent the attacker from reaching critical assets. We introduced a method to maintain a belief about the attacker's progression in a distributed manner, facilitating timely blocking and proactive threat response. Our numerical results demonstrate that the proposed approach significantly reduces the total cost while efficiently managing communication overhead. 
\revision{While the central focus of this work is on the decision-theoretic framework and the distributed belief update rather than engineering a practical detection tool, our evaluation is grounded in  numerical simulations conforming to the MITRE ATT\&CK matrix. Validating the framework against large-scale operational traces and heterogeneous real-world testbeds remains an important direction for future work. In addition to a real world deployment, there are several promising directions for theoretical and structural future research.}
First, one could explore malware that evolves by interacting with compromised hosts. Second, the framework could be extended to account for multiple simultaneous actions by the attacker, where each compromised machine can independently target other subnetworks. Third, incorporating strategic attackers that dynamically adapt their strategy based on the network environment and its belief about the attacker´s progression.
\vspace{-5mm}
\bibliography{biblio.bib}
\vspace{-9mm}
\begin{IEEEbiography}[{\includegraphics[width=1in,height=1.25in,clip,keepaspectratio]{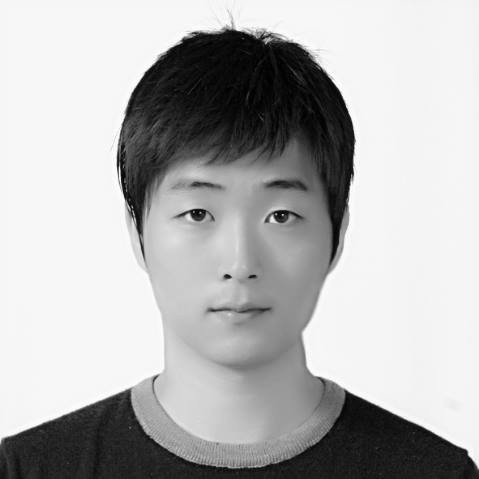}}]%
{Yeongwoo Kim} received his M.Sc. degree in Electrical Engineering from both KTH Royal Institute of Technology, Sweden, and the Technical University of Berlin, Germany, in 2020. He is currently pursuing a Ph.D. in the Division of Network and Systems Engineering at KTH Royal Institute of Technology in Stockholm, Sweden. His research interests icnludes human-in-the-loop systems, machine learning algorithms, and analytic approaches for cybersecurity.
\end{IEEEbiography}\vspace{-9mm}
\begin{IEEEbiography}[{\includegraphics[width=1in,height=1.25in,clip,keepaspectratio]{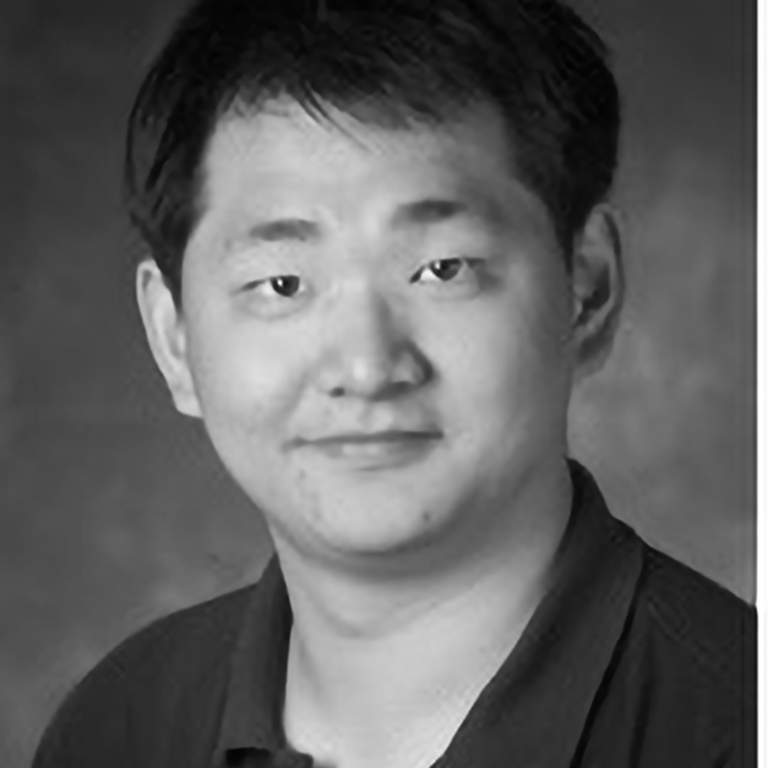}}]%
{Quanyan Zhu} Quanyan Zhu (Senior Member, IEEE) received the B.Eng. degree (Hons.) in electrical engineering from McGill University in 2006, the M.A.Sc. degree from the University of Toronto in 2008, and the Ph.D. degree from the University of Illinois at Urbana-Champaign (UIUC) in 2013. After stints at Princeton University, he is currently an Associate Professor with the Department of Electrical and Computer Engineering, New York University (NYU), where he is an Affiliated Faculty Member with the Center for Urban Science and Progress (CUSP). He is the coauthor of two recent books published by Springer: \textit{Cyber-Security in Critical Infrastructures: A Game-Theoretic Approach} (with S. Rass, S. Schauer, and S. König) and \textit{A Game- and Decision-Theoretic Approach to Resilient Interdependent Network Analysis and Design} (with J. Chen). His current research interests include game theory, machine learning, cyber deception, network optimization and control, smart cities, the Internet of Things, and cyber-physical systems. He was a recipient of many awards, including the NSF CAREER Award, the NYU Goddard Junior Faculty Fellowship, the NSERC Postdoctoral Fellowship (PDF), the NSERC Canada Graduate Scholarship (CGS), and the Mavis Future Faculty Fellowships. He spearheaded and chaired the INFOCOM Workshop on Communications and Control on Smart Energy Systems (CCSES), the Midwest Workshop on Control and Game Theory (WCGT), and the ICRA workshop on Security and Privacy of Robotics. He served as the General Chair or the TPC Chair for the Seventh and the 11th Conference on Decision and Game Theory for Security (GameSec) in 2016 and 2020, the Ninth International Conference on NETwork Games, COntrol and OPtimization (NETGCOOP) in 2018, the Fifth International Conference on Artificial Intelligence and Security (ICAIS 2019) in 2019, and the 2020 IEEE Workshop on Information Forensics and Security (WIFS). He also spearheaded the IEEE Control System Society (CSS) Technical Committee on Security, Privacy, and Resilience, in 2020.
\end{IEEEbiography}\vspace{-9mm}
\begin{IEEEbiography}[{\includegraphics[width=1in,height=1.25in,clip,keepaspectratio]{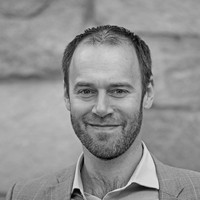}}]%
{Gy\"orgy D\'an} (Senior Member, IEEE) received the M.Sc. degree in computer engineering from the Budapest University of Technology and Economics, Hungary, in 1999, the M.Sc. degree in business administration from the Corvinus University of Budapest, Hungary, in 2003, and the Ph.D. degree in telecommunications from KTH in 2006. He was a Consultant in the field of access networks, streaming media, and videoconferencing from 1999 to 2001. He was a Visiting Researcher at the Swedish Institute of Computer Science in 2008, a Fulbright Research Scholar at the University of Illinois at Urbana–Champaign from 2012 to 2013, and an Invited Professor at EPFL in from 2014 to 2015. He is a Professor with the KTH Royal Institute of Technology, Stockholm, Sweden. His research interests include the design and analysis of content management and computing systems, game theoretical models of networked systems, and cyber-physical system security and resilience. He was area editor of Computer Communications from 2014 to 2021 and of IEEE Trans. on Mobile Computing 2019-2023.
\end{IEEEbiography}

\begin{appendices}
\section*{Appendix: Belief estimation based on diffusion HMM}\label{sec:diffusion}
\begin{figure}[!t]
\begin{center}
    \includegraphics[width = \linewidth]{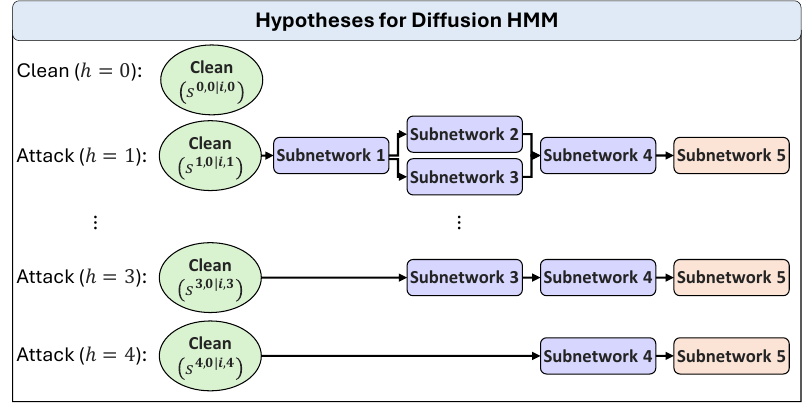}
      \caption{Hypotheses in subnetwork $i$ for the diffusion HMM. Each oval and rectangle represents a state and a subnetwork, respectively. Each subnetwork contains three states: access, reconnaissance, and lateral movement. Attack hypothesis $h$ corresponds to an attacker initiating access from subnetwork $h$.}\label{fig:diffusion_hypotheses}
\end{center}
\end{figure}

We adopt the distributed belief update approach from~\cite{kayaalp2022hidden}. We modify~\cite{kayaalp2022hidden} by allowing a) that each subnetwork maintains up to  $\vert \iSet \setminus\iSet^c \vert+1$ hypotheses, while the diffusion HMM uses a single hypothesis, and b) that each subnetwork maintains likelihoods while the diffusion HMM maintains belief. 

Fig.~\ref{fig:diffusion_hypotheses} shows $\vert \iSet\setminus\iSet^c \vert+1$ hypotheses considered in subnetwork $i$ to model the random initial compromise of an attack, along with a clean hypothesis representing a fully clean network. Let us denote by $\hypothesisSet = \{0\}\cup \iSet\setminus\iSet^c$ the set of hypotheses, where $\hypothesis \in \hypothesisSet$, by $\stage^{i^\prime,j^\prime \vert i,h}_t$ the state $s^{i^\prime,j^\prime}$ in the kill chain of subnetwork $i$ under hypothesis $h$, by $\alpha^{i^\prime,j^\prime \vert i,h }_{t,d}$ the likelihood of state $\stage^{i^\prime,j^\prime \vert i,h}_t$ maintained by the diffusion HMM, and $\belief^{i^\prime,j^\prime \vert i,h }_{t,d}$ the corresponding belief.
The clean hypothesis $\hypothesis=0$ consists of a single benign state $\{\stage^{0,0\vert i, 0}\}$, and the attack hypothesis $h \in \hSet^i$ consists of the set of states $\{\stage^{0,0\vert i, 0}\} \cup_{i \in \iSet^{h}} \mathcal{S}^i$. 
An attack hypothesis $h >0$ corresponds to an initial compromise originating from subnetwork $h$ ($\tranProb_t^{0,0 \rightarrow h,1 \vert i,h} > 0$, while $\tranProb_t^{0,0 \rightarrow i^{\prime},j \vert i, h} = 0$ for all $i^{\prime} \in \iSet \setminus \{h\}$ and $j \in \jSet$).
Each subnetwork updates its belief using transition probabilities:
\begin{align}
   \acute{\alpha}^{i^\prime,j^\prime\vert i, h}_{t,d}=
   \begin{cases}
   \alpha^{0,0 \vert i, h}_{t-1,d}\tranProb^{0,0  \rightarrow 0,0 \vert i, h},\\ \qquad \qquad \qquad \qquad \text{for } i^\prime = 0,  j^\prime=0,\\
   \alpha^{0,0 \vert i, h}_{t-1,d}\tranProb^{0,0  \rightarrow i^\prime,1 \vert i, h} + \alpha^{i^\prime,1 \vert i, h}_{t-1,d}\tranProb^{i^\prime,1 \rightarrow i^\prime,1},\\ \qquad \qquad \qquad \qquad \text{for } i^\prime = h, j^\prime=1,\\
   \sum_{i^{\prime\prime} \in \iSet^h}\sum_{j^{\prime\prime} \in \jSet^{i^{\prime\prime}}} \alpha^{i^{\prime\prime},j^{\prime\prime} \vert h}_{t-1,d}\tranProb_t^{i^{\prime\prime},j^{\prime\prime} \rightarrow i^\prime,j^\prime}, \\ \qquad \qquad \qquad \qquad \text{otherwise},
   \end{cases}
\end{align}
where $\alpha^{i^{\prime},j^{\prime} \vert i, h}_{t-1,d}$ is the likelihood of state $s^{i^{\prime},j^{\prime} \vert i, h}$ at time $t-1$, and $\acute{\alpha}^{i^{\prime},j^{\prime} \vert i,h}_{t,d}$ captures the transition-evolved likelihood. 
Next, each subnetwork incorporates locally observed alerts to compute the state likelihood:
\begin{align}
    \hat{\alpha}^{i^{\prime},j^{\prime}\vert i,h}_{t,d}= \P\left(\alertVectorRv^{i}_t=\alertVector^{i}_t\vert \stageRv_t=\stage^{i^{\prime},j^{\prime}}\right)^\kappa \acute{\alpha}^{i^{\prime},j^{\prime}\vert i,h}_{t,d},
\end{align}
where $\kappa \in \mathbbm{N}$.

To enable consensus across subnetworks, we introduce a symmetric likelihood-sharing matrix $\mathbf{B}$ satisfying $\mathbf{B}\cdot\mathbf{1} = \mathbf{1}$ and $\mathbf{B}^T = \mathbf{B}$. Let $b^{i,i^{\prime}}$ denote the $(i,i^\prime)$ entry of the belief-sharing matrix $\mathbf{B}$, representing the weight that subnetwork $i$ assigns to the local likelihoods of subnetwork $i^{\prime}$. The updated likelihood is computed as:
\begin{align}
   \alpha^{i^{\prime},j^{\prime} \vert i, h}_{t,d} &=
       \exp\left(\sum_{i^{\prime\prime} \in \iSet^h} b^{i,i^\prime} \ln\left(\hat{\alpha}^{i^{\prime}, j^{\prime} \vert i^{\prime\prime}, h}_{t,d}\right)\right).\label{eq:diffusion_geomean}
\end{align}
We normalize the aggregated likelihood to obtain the belief distribution as follows:
\begin{align}
   \belief^{i^\prime,j^\prime\vert i,h}_{t,d} = \frac{ \alpha^{i^\prime,j^\prime \vert i, h}_{t,d}}{\alpha^{0,0 \vert i, h}_{t,d} +\sum_{i^\prime\in \iSet^h}\sum_{j^\prime\in \mathcal{J}}\alpha^{i^\prime,j^\prime \vert i, h}_{t,d}}.
\end{align}
resulting belief distribution is then used to determine which subnetworks to block (see Section~\ref{sec:threat_blocking}).

\section*{Appendix: KL Divergence Bound Between Centralized and {\Chain} Beliefs}\label{sec:proof_chain}

\begin{figure}[!t]
\begin{center}
    \includegraphics[width = \linewidth]{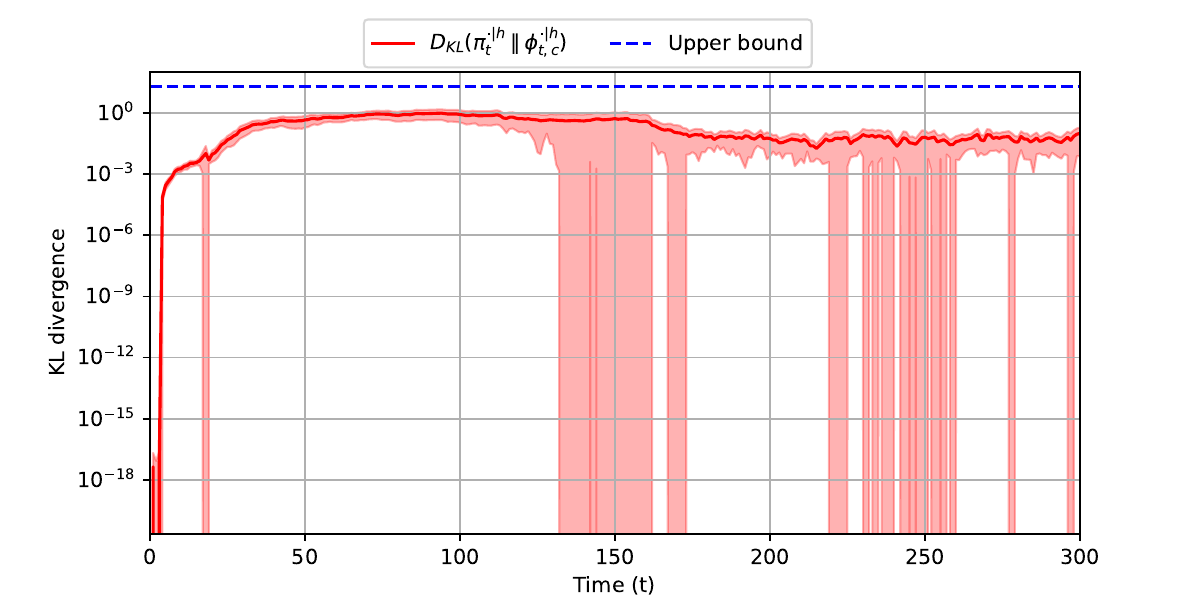}
      \caption{Average KL divergence between central and partitioned belief filtering with 95\% confidence intervals. The simulation is conducted without blocking or termination.}\label{fig:KL_div}
\end{center}
\end{figure}
We evaluate the discrepancy between the centralized belief $\pi^{\cdot,\cdot \vert h}_t$ and the {\chain} belief $\belief^{\cdot,\cdot \vert h}_{t,p}$ using the \ac{KL} divergence. 
\begin{align}
    D_{KL}(\belief^{\cdot,\cdot \vert h}_t \parallel \belief^{\cdot,\cdot \vert h }_{t,p})&= \sum_{i\in \iSet^h}\sum_{j\in\jSet^i} \belief^{i,j \vert h}_t \log \left(\frac{\belief^{i,j \vert h}_t }{\belief^{i,j \vert h }_{t,p}}\right).\label{eq:KL_div}
\end{align}
Plugging~\eqref{eq:chain_observation} and~\eqref{eq:chain_normalize} into~\eqref{eq:chain_central}, we can rewrite it as:
\begin{align}
    {\belief}^{i,j\vert h}_{t,p}
    &={\eta}^h_{t,p} \P(\alertVectorRv^{\iSet}_t = \alertVector^{\iSet}_t \vert \stageRv_t = \stage^{i,j})\acute{\alpha}^{i,j\vert h}_{t,p}\prod_{i^\prime\in\iSet^h}\eta^{i^\prime,h}_{t,p} \\&\quad\cdot\hspace{-3mm}\prod_{i^{\prime\prime} \in \iSet^h \cap \iSet^{\uparrow}(i)}
    \hspace{-3mm}\acute{\alpha}^{i^{\prime\prime},J+1\vert h}_{t,p}
    \hspace{-3mm}\prod_{i^{\prime\prime\prime} \in \iSet^h \setminus \iSet^{\uparrow}(i)}\hspace{-3mm}\acute{\alpha}^{i^{\prime\prime\prime},0\vert h}_{t,p}.\label{eq:chain_central_detailed}
\end{align}
where ${\eta}^h_{t,p}$ is expressed as
\begin{align}
    {\eta}^h_{t,p}\hspace{-1mm} = \hspace{-1mm}\frac{1}{\begin{array}{l}
        \bigg(\P(\alertVectorRv^{\iSet}_t = \alertVector^{\iSet}_t \vert \stageRv_t = \stage^{i,0}) \acute{\alpha}^{i,0\vert h}_{t,p} \\ \qquad +\sum_{i\in\iSet^h}\sum_{j\in\jSet}\P(\alertVectorRv^{\iSet}_t = \alertVector^{\iSet}_t \vert \stageRv_t = \stage^{i,j}) \acute{\alpha}^{i,j\vert h}_{t,p} \bigg)
        \\ \qquad \cdot \prod_{i^\prime\in\iSet^h}\eta^{i^\prime,h}_{t,p} \prod_{i^{\prime\prime} \in \iSet^h \cap \iSet^{\uparrow}(i)}\acute{\alpha}^{i^{\prime\prime},J+1\vert h}_{t,p}
        \\ \qquad \cdot \prod_{i^{\prime\prime\prime} \in \iSet^h \setminus \iSet^{\uparrow}(i)}\acute{\alpha}^{i^{\prime\prime\prime},0\vert h}_{t,p}
    \end{array}},
\end{align}
where $\eta^{i^\prime,h}_{t,p}$ is the normalization factor used in~\eqref{eq:chain_normalize}.
Plugging~\eqref{eq:chain_central_detailed} into~\eqref{eq:KL_div}, we derive:
\begin{align}
   &D_{KL}(\belief^{\cdot,\cdot \vert h}_t \parallel \belief^{\cdot,\cdot \vert h }_{t,p})\\&=\sum_{i\in \iSet^h}\hspace{-1mm}\sum_{j\in\jSet^i} \hspace{-1mm}\belief^{i,j \vert h}_t \log \hspace{-1mm} \left(\hspace{-1.5mm}\frac{ \alpha^{i,j}_{t} }{\splitfrac{\acute{\alpha}^{i,j\vert h}_{t,p} \prod_{i^{\prime\prime} \in \iSet^h \cap \iSet^{\uparrow}(i)}\acute{\alpha}^{i^{\prime\prime},J+1\vert h}_{t,p}}{\prod_{i^{\prime\prime\prime} \in \iSet^h \setminus \iSet^{\uparrow}(i)}\acute{\alpha}^{i^{\prime\prime\prime},0\vert h}_{t,p}}}\hspace{-1.5mm}\right)\\
    &\qquad + \log\left(\frac{\eta_t}{{\eta}^{h}_{t,p}\prod_{i^{\prime}\in\iSet^h}\eta^{i^{\prime},h}_{t,p}}\right),\\
    &=\sum_{i\in \iSet^h}\hspace{-1mm}\sum_{j\in\jSet^i} \belief^{i,j \vert h}_t \log \hspace{-1mm}\left(\hspace{-1.5mm}\frac{ \alpha^{i,j}_{t} }{\splitfrac{\acute{\alpha}^{i,j\vert h}_{t,p} \prod_{i^{\prime\prime} \in \iSet^h \cap \iSet^{\uparrow}(i)}\acute{\alpha}^{i^{\prime\prime},J+1\vert  h}_{t,p}}{\prod_{i^{\prime\prime\prime} \in \iSet^h \setminus \iSet^{\uparrow}(i)}\acute{\alpha}^{i^{\prime\prime\prime},0\vert  h}_{t,p}}}\hspace{-1.5mm}\right)\\
    &\qquad+\log(\eta_t)
    \\&\qquad+\log\bigg(\big(\P(\alertVectorRv^{\iSet}_t = \alertVector^{\iSet}_t \vert \stageRv_t = \stage^{i,0}) \acute{\alpha}^{i,0\vert h}_{t,p}+ \\&\qquad\qquad+\sum_{i\in\iSet^h}\sum_{j\in\jSet}\P(\alertVectorRv^{\iSet}_t = \alertVector^{\iSet}_t \vert \stageRv_t = \stage^{i,j})\acute{\alpha}^{i,j\vert  h}_{t,p} \big) \\&\qquad\qquad\prod_{i^{\prime\prime} \in \iSet^h \cap \iSet^{\uparrow}(i)}\acute{\alpha}^{i^{\prime\prime},J+1\vert  h}_{t,p}
    \prod_{i^{\prime\prime\prime} \in \iSet^h \setminus \iSet^{\uparrow}(i)}\acute{\alpha}^{i^{\prime\prime\prime},0\vert  h}_{t,p}\bigg),
\end{align}

We observe that the last two terms are non-positive, since the arguments of the logarithm are bounded by one. Consequently, we obtain an upper bound by omitting these terms, which yields inequality (a). Next, we apply the arithmetic–geometric mean inequality to further relax the expression, leading to inequality (b) as follows:
\begin{align}
    &D_{KL}(\belief^{\cdot,\cdot \vert h}_t \parallel \belief^{\cdot,\cdot \vert h }_{t,p})\\&\myleqA \hspace{-1mm}\sum_{i\in \iSet^h}\sum_{j\in\jSet^i} \belief^{i,j \vert h}_t \log \hspace{-1mm}\left(\hspace{-1mm}\frac{ \alpha^{i,j}_{t} }{\splitfrac{\acute{\alpha}^{i,j\vert h}_{t,p} \prod_{i^{\prime\prime} \in \iSet^h \cap \iSet^{\uparrow}(i)}\acute{\alpha}^{i^{\prime\prime},J+1\vert  h}_{t,p}}{\prod_{i^{\prime\prime\prime} \in \iSet^h \setminus \iSet^{\uparrow}(i)}\acute{\alpha}^{i^{\prime\prime\prime},0\vert  h}_{t,p}}}\hspace{-1mm}\right)\\
    &\myleqB  \hspace{-1mm}\log \hspace{-1mm}\left(\hspace{-1mm}\sum_{i\in \iSet^h}\sum_{j\in\jSet^i} \belief^{i,j \vert h}_t\frac{ \alpha^{i,j}_{t} }{\splitfrac{\acute{\alpha}^{i,j\vert  h}_{t,p} \prod_{i^{\prime\prime} \in \iSet^h \cap \iSet^{\uparrow}(i)}\acute{\alpha}^{i^{\prime\prime},J+1\vert  h}_{t,p}}{\prod_{i^{\prime\prime\prime} \in \iSet^h \setminus \iSet^{\uparrow}(i)}\acute{\alpha}^{i^{\prime\prime\prime},0\vert  h}_{t,p}}}\hspace{-1mm}\right)\hspace{-1mm}.
\end{align}
This indicates that the upper bound of the KL divergence at time $t$ depends solely on the belief at time $t-1$, independent of the alert observations. By applying~\eqref{eq:alpha_central} and~\eqref{eq:chain_prior}, we can reformulate the above expression as follows:

\begin{align}
    &D_{KL}(\belief^{\cdot,\cdot \vert h}_t \parallel \belief^{\cdot,\cdot \vert h }_{t,p})\\&\leq
    \log \Bigg(\sum_{i\in \iSet^h}\sum_{j\in\jSet^i} \belief^{i,j \vert h}_{t}\frac{ \sum_{i^{\prime} \in \iSet, j^\prime\in \jSet} \belief^{i^\prime,j^\prime \vert h}_{t-1} \cdot  \tranProb^{i^{\prime},j^\prime \rightarrow i,j}_{t-1}}{\beta_{t-1}}\Bigg).
\label{eq:KL_div_upperbound1}
\end{align}
where 
\begin{align}
\begin{aligned}
    \beta_{t-1}=&\sum_{j^{\prime} \in \jSet^{i \vert h}}\left(\distributedBelief^{i,j^{\prime} \vert h}_{t-1,p} \cdot  \tranProb^{i,j^{\prime} \rightarrow i,j \vert h}_{t-1}\right)\\
    &\cdot\hspace{-5mm}\prod_{i^{\prime\prime} \in \iSet^h \cap \iSet^{\uparrow}(i)}\hspace{-5mm}\left(\distributedBelief^{i^{\prime\prime},J \vert h}_{t-1,p} \cdot  \tranProb^{i^{\prime\prime},J \rightarrow i^{\prime\prime},J+1 \vert  h}_{t-1} + \distributedBelief^{i^{\prime\prime},J+1 \vert  h}_{t-1,p} \right)\\
    &\cdot\hspace{-5mm}\prod_{i^{\prime\prime\prime} \in \iSet^h \setminus \iSet^{\uparrow}(i)}\hspace{-5mm}\left(\distributedBelief^{i^{\prime\prime\prime},0 \vert h}_{t-1,p} \cdot(1-  \tranProb^{i^{\prime\prime\prime},0 \rightarrow i^{\prime\prime\prime},1 \vert  h}_{t-1})\right).\end{aligned}
\end{align}
The upperbound of right-hand side of~\eqref{eq:KL_div_upperbound1} can be expressed as, 
\begin{align}
    &D_{KL}(\belief^{\cdot,\cdot \vert h}_t \parallel \belief^{\cdot,\cdot \vert h }_{t,p})\\&\leq
    \log \Bigg(\sum_{i\in \iSet^h}\sum_{j\in\jSet^i} \belief^{i,j \vert h}_t\frac{ \sum_{i^{\prime} \in \iSet, j^\prime\in \jSet} \belief^{i^\prime,j^\prime \vert h}_{t-1} \cdot  \tranProb^{i^{\prime},j^\prime \rightarrow i,j}_{t-1}}{\beta_{t-1}}\Bigg)\\&\leq
    \log \Bigg(\max\frac{ \sum_{i^{\prime} \in \iSet, j^\prime\in \jSet} \belief^{i^\prime,j^\prime}_{t-1} \cdot  \tranProb^{i^{\prime},j^\prime \rightarrow i,j}_{t-1}}{\beta_{t-1}}\Bigg).
\label{eq:KL_div_upperbound2}
\end{align}
Recalling assumptions~\ref{assume:belief_lowerbound} and~\ref{assume:Bounded_transition}, we proceed by maximizing the numerator and minimizing the denominator. The maximum value of the numerator is  $q_{\text{max}}$. We observe that: i) The minimum of $1-\tranProb^{i,0\rightarrow i,1 \vert h}_{t-1}$ is given by $\min\left(1- \sum_{i^\prime \in \iSet^h \cap \iSet^{\uparrow}_d(i)}w^{i^\prime,i\vert h}\cdot 1\cdot q_{\text{max}} \right)= 1-q_{\text{max}}$ since the weights sum to one. ii) Consequently, we define $q_{\text{min}}=1- q_{\text{max}}$ based on the assumption that each state transitions either to itself or to the next state. Therefore, the transition probability satisfies $\tranProb^{i,j^{\prime} \rightarrow i,j \vert h}_{t-1} \in [q_{\text{min}}, q_{\text{max}}]$. Using these bounds, we now rewrite the denominator as follows:
\begin{align}
\begin{split}
    \beta_{t-1}=&\sum_{j^{\prime} \in \jSet^{i \vert h}}\left(\epsilon^2 q_{\text{min}}\right) \hspace{-5mm}\prod_{i^{\prime\prime} \in \iSet^h \cap \iSet^{\uparrow}(i)}\hspace{-5mm}\left(\epsilon \cdot  q_{\text{min}}+ \epsilon \right)\hspace{-5mm}\prod_{i^{\prime\prime\prime} \in \iSet^h \setminus \iSet^{\uparrow}(i)}\hspace{-5mm}\left(\epsilon \cdot(1-  q_{\text{max}})\right),
    \end{split}\\
\begin{split}
    =&\sum_{j^{\prime} \in \jSet^{i \vert h}}\left(\epsilon^2 q_{\text{min}}\right)\hspace{-5mm} \prod_{i^{\prime\prime} \in \iSet^h \cap \iSet^{\uparrow}(i)}\hspace{-5mm}\left(\epsilon \cdot  q_{\text{min}}+ \epsilon \right)\hspace{-5mm}\prod_{i^{\prime\prime\prime} \in \iSet^h \setminus \iSet^{\uparrow}(i)}\hspace{-5mm}\left(\epsilon q_{\text{min}})\right),
    \end{split}\\
    \begin{split}
        \geq&\sum_{j^{\prime} \in \jSet^{i \vert h}}\left(\epsilon^2 q_{\text{min}}\right) \hspace{-5mm}\prod_{i^{\prime\prime} \in \iSet^h \cap \iSet^{\uparrow}(i)}\hspace{-5mm}\left(\epsilon   q_{\text{min}} \right)\hspace{-5mm}\prod_{i^{\prime\prime\prime} \in \iSet^h \setminus \iSet^{\uparrow}(i)}\hspace{-5mm}\left(\epsilon q_{\text{min}})\right),
    \end{split}\\
    \begin{split}
        =&\sum_{j^{\prime} \in \jSet^{i \vert h}} \epsilon^{\vert \iSet^h \vert +1}   q_{\text{min}}^{\vert \iSet^h \vert}.
    \end{split}
\end{align}
Therefore, the upperbound of the KL divergence is 
\begin{align}
    D_{KL}(\belief^{\cdot,\cdot \vert h}_t \parallel \belief^{\cdot,\cdot \vert h }_{t,p})&\leq
    \log \Bigg(\frac{q_{\text{max}}}{\sum_{j^{\prime} \in \jSet^{i \vert h}} \epsilon^{\vert \iSet^h \vert +1}   q_{\text{min}}^{\vert \iSet^h \vert}}\Bigg).
\label{eq:KL_div_upperbound3}
\end{align}
Let $\epsilon =0.001$. We conduct 100 simulations without applying blocking actions to evaluate the average KL divergence between the central belief and the {\chain} belief. Fig.~\ref{fig:KL_div} shows the KL divergence as a function of time. This figure shows that the mean KL divergence remains within the theoretical upper bounds, demonstrating that the proposed belief update mechanism exhibits stable performance and converges toward the central belief by time $t=300$.  
\section*{Appendix: Transformer-based Autoencoder}\label{sec:transformer}

The transformer‑based autoencoder architecture consists of an embedding layer, a stack of transformer encoder layers, and a final expansion layer. Let $n_{likelihood}$ denote the input likelihood dimensionality and $n_{embed}$ the embedding size. The embedding layer applies a linear transformation to compress $n_{likelihood}$ into $n_{embed}$, incorporating a causal mask to ensure reconstructions at time $t$ exclude future inputs $t'>t$, and a positional encoder to preserve temporal order. The transformer employs 4 encoder layers with 8 multi-head attention heads, 0.1 dropout, and feed-forward dimension $2n_{embed}$ to capture temporal correlations. The expansion layer linearly maps back to $n_{likelihood}$. Input dimensionality varies by belief mode: $n_{likelihood}=85$ for \emph{Centralized} and $85 \cdot I$ for \emph{Partitioned}. Grid search over $n_{embed} \in \{8,16,32,64\}$ identifies $n_{embed}=16$ as optimal.

The input to the model consists of the likelihoods of all hypotheses and subnetworks. Because these likelihoods are computed by repeatedly multiplying values smaller than 1 at each time step, they naturally decay toward zero over the course of a simulation. To prevent this collapse and avoid trivial zero‑valued predictions near the end of an episode, the likelihoods are normalized at every time step. The model is trained to reconstruct these normalized likelihoods, using KL divergence as the reconstruction loss to reflect the probabilistic structure of the inputs.

We simulate 100 episodes without attack and divide them into training and validation sets using a 90/10 split. The model is trained with early stopping, which terminates training when the validation performance fails to improve for more than five consecutive evaluations. This entire procedure from generating new episodes to training the model is repeated ten times to obtain a robust model. 

\end{appendices}

\end{document}